\documentclass[
 reprint,
superscriptaddress,
 amsmath,amssymb,
 aps,
pra,
longbibliography
]{revtex4-1}

\usepackage{graphicx}
\usepackage{dcolumn}
\usepackage{bm}
\usepackage{bbm}
\usepackage[utf8]{inputenc}
\usepackage{amsfonts}
\usepackage{amsthm}
\usepackage{tikz}
\usepackage{physics}
\usepackage{comment}
\usepackage{xcolor}
\usepackage{blindtext}
\usepackage{pgfplots,pgfplotstable}
\usepgfplotslibrary{polar}
\usepgfplotslibrary{colormaps}
\pgfplotsset{compat=newest}
\usepackage{hyperref}

\definecolor{DTgreen}{rgb}{.1, .5, .2}

\definecolor{HS}{rgb}{0.0,0.5,0.0}
\definecolor{SN}{rgb}{1, 0.5, 0.31}
\definecolor{SSN}{rgb}{0.87, 1.0, 0.0}
\definecolor{OSN}{rgb}{1.0, 0.44, 0.37}
\newcommand{\re}[1]{\Re[#1]}
\newcommand{\im}[1]{\Im[#1]}

\providecommand{\R}{\mathbb{R}}

\providecommand{\C}{\mathbb{C}}
\providecommand{\E}{\mathbb{E}}

\providecommand{\abs}[1]{\left\lvert#1\right\rvert}
\providecommand{\Norm}[1]{\left \lVert #1 \right\rVert}
\newcommand{\ud}{\,\mathrm{d}}
\renewcommand{\d}{\mathrm{d}}
\newcommand{\sd}{\overset{d}{=}}
\newcommand{\e}{\mathrm{e}}
\newcommand{\ii}{\mathrm{i}}
\newcommand{\diag}{\mathrm{diag}}
\newcommand{\Uphi}{U_\varphi}
\newcommand{\hUphi}{\hat{U}_\varphi}
\newcommand{\uphi}{u_\varphi}
\newcommand{\huphi}{\hat{u}_\varphi}
\newcommand{\uphiff}{(\uphi)_{11}}
\newcommand{\duphiff}{\partial_\varphi \uphiff}
\newcommand{\Vi}{V_{\mathrm{in}}}
\newcommand{\Vo}{V_{\mathrm{out}}}
\newcommand{\hVi}{\hat{V}_{\mathrm{in}}}
\newcommand{\hVo}{\hat{V}_{\mathrm{out}}}

\newcommand{\f}{f(U,G_\varphi)}

\theoremstyle{theorem}
\newtheorem{thm}{Theorem}[]

\newtheorem{lemma}{Lemma}[]
\theoremstyle{definition}

\begin{document}

\title{Typicality of Heisenberg scaling precision in multi-mode quantum metrology}

\author{Giovanni Gramegna}
\email{giovanni.gramegna@ba.infn.it}
\affiliation{Dipartimento di Fisica and MECENAS, Università di Bari, I-70126 Bari, Italy }
\affiliation{INFN, Sezione di Bari, I-70126 Bari, Italy}

\author{Danilo Triggiani}
\email{danilo.triggiani@port.ac.uk}
\affiliation{School of Mathematics and Physics, University of Portsmouth, Portsmouth PO1 3QL, UK}

\author{Paolo Facchi}
\affiliation{Dipartimento di Fisica and MECENAS, Università di Bari, I-70126 Bari, Italy }
\affiliation{INFN, Sezione di Bari, I-70126 Bari, Italy}

\author{Frank A. Narducci}
\affiliation{Department of Physics, Naval Postgraduate School, Monterey, CA, United States 
}

\author{Vincenzo Tamma}
\email{vincenzo.tamma@port.ac.uk}
\affiliation{School of Mathematics and Physics, University of Portsmouth, Portsmouth PO1 3QL, UK}
\affiliation{Institute of Cosmology and Gravitation, University of Portsmouth, Portsmouth PO1 3FX, UK}

\date{\today}

\begin{abstract}
We propose a measurement setup reaching Heisenberg scaling precision for the estimation of any distributed parameter $\varphi$ (not necessarily a phase) encoded into a generic $M$-port linear network composed only of passive elements. The scheme proposed can be easily implemented from an experimental point of view since it employs only Gaussian states and Gaussian measurements. Due to the complete generality of the estimation problem considered, it was predicted that one would need to carry out an adaptive procedure which involves both the input states employed and the measurement performed at the output; we show that this is not necessary: Heisenberg scaling precision is still achievable by only adapting a single stage. The non-adapted stage only affects the value of a pre-factor multiplying the Heisenberg scaling precision: we show that, for large values of $M$ and a random (unbiased) choice of the non-adapted stage, this pre-factor takes a typical value which can be controlled through the encoding of the parameter $\varphi$ into the linear network. 
\end{abstract}

\pacs{Valid PACS appear here}
\maketitle

\section{Introduction}
The precision achievable in a measurement when all experimental noise sources are minimized is ultimately determined by the discreteness of all physical phenomena: electronic devices will suffer the discreteness of the electric charge, whereas the quantum nature of light will affect optical devices. Due to this quantum noise, 
the error in the estimation of a physical parameter $\varphi$ through a measurement employing $N$ probes (e.g. photons, electrons) is strongly limited by the so-called ``shot noise'' factor of $1/\sqrt{N}$. However, it has been proven that quantum features such as entanglement and squeezing can be exploited to go beyond the shot-noise limit  and reach a precision of order $1/N$, the so-called Heisenberg limit (HL)~\cite{Giovannetti2004,Giovannetti2006,Dowling2008,Giovannetti2011,Dowling2015,Shapiro84,Wineland92,Maccone2019}. 

Several quantum metrological problems have been largely studied and a few approaches have been proposed to reach a HL sensitivity. However, these protocols are usually difficult to implement experimentally due to the convoluted and challenging measurement procedures~\cite{Helstrom1969,PhysRevLett.72.3439,Seshadreesan_2011} and the fragile quantum coherence needed in the input states~\cite{Giovannetti2006,PhysRevA.92.042115,PhysRevA.76.013804,PhysRevA.90.025802}. 
Gaussian states, on the other hand, provide a promising avenue for quantum optical technologies~\cite{Weedbrook2012,Adesso2014}, since they are easier to create and manipulate experimentally compared to non-Gaussian ones, such as Fock states. Moreover, they allow a complete analytical treatment from a theoretical point of view~\cite{Paris2005,Weedbrook2012,Adesso2014}. In particular, the \textit{squeezing} of a Gaussian state, which allows for highly reduced-noise signals, appears to be a valuable tool to reach quantum super-sensitive precision~\cite{Maccone2019}. From a metrological perspective, squeezed states are often used along with Gaussian measurements~\cite{monras2006,Matsubara_2019,Oh2019}, defined as measurement schemes producing a Gaussian probability distribution of the outcomes for any Gaussian state~\cite{Weedbrook2012}. Homodyne and heterodyne detection represent paradigmatic examples of Gaussian measurements. It has been shown, both theoretically~\cite{adaptiveHomodyne1995} and experimentally~\cite{adaptiveHomodyne2002}, that an adaptive homodyne phase estimation performs better than heterodyne detection, and approaches closer to the intrinsic quantum uncertainty than any previous technique when no prior knowledge of the phase is given. The importance of feedback and adaptivity in quantum estimation protocols has been underlined also in subsequent works~\cite{monras2006,Aspachs2009}. Adaptiveness can be avoided in an optimal protocol (or near optimal) only if some constraint in the range of variation of the parameter is given~\cite{Gatto2019,Gatto2020,PhysRevA.96.052118}.

Within the domain of quantum optics, photons are sent as probes through an interferometer where a parameter $\varphi$ to be estimated is encoded. The information about the parameter is imprinted then in the output state of the photons, and it can be extracted by a suitable measurement. The situation which has been often considered is the case where $\varphi$ is an optical phase~\cite{Shapiro84,Giovannetti2004,monras2006,Dowling2008,Oh2019} or a phase-like parameter~\cite{Giovannetti2006,Giovannetti2011}. These results clearly apply also to situations in which other quantities of interest (e.g.\ a distance) can be converted into an optical phase~\cite{Dowling2008}, but they fail to cover more general situations, e.g.\ the unknown parameter is distributed among several components of the interferometer. Recently, some progress has been made along this direction concerning the estimation of particular functions of multiple parameters distributed in a specific manner within a particular network. \cite{Ge2018,Zhuang2018,Sidhu2020,Qian2019,Xia2020,Guo2020,Nair2018}. It has been also shown in a recent work~\cite{Matsubara_2019} that the presence of a single unknown parameter distributed in multiple nodes of an arbitrary network introduces non-trivial complications if no constraints are given on the range of values the parameter is allowed to assume: in fact, it appears that a simultaneous adaptive procedure both in the input probe and in the measurement is needed in order to reach the HL, making the whole scheme quite unfeasible from a practical point of view. Furthermore, the proposed scheme requires an unquantified precision and number of resources  in the adaptive procedure.

In this work we demonstrate the typicality of the Heisenberg limited sensitivity with a simple metrological technique which overcomes at the same time all these serious drawbacks. In particular, we consider a general scenario in which $\varphi$ can be any parameter embedded into an arbitrary linear passive $M$-modes interferometer: it can be a parameter characterizing any specific component of the interferometer, or arbitrarily distributed among different components of the circuit. We will show that an experimentally feasible scheme achieving Heisenberg scaling is typically possible in such a general scenario. 
In our scheme (see~\figurename\ref{fig:Generic setup}), a single-mode squeezed vacuum state is sent through a linear, passive preliminary stage which scatters the input photons among all the $M$ channels of the interferometer, in order to extract the information on $\varphi$ which is distributed among all the modes. A second auxiliary stage at the output of the interferometer refocuses the photons in the only observed output port. By employing a single-mode homodyne detection, we present two broad conditions which together suffice to reach the HL: the first being the requirement that most of the injected photons are successfully refocused on the observed output mode; the second simply being a minimal-resolution requirement on the homodyne measurement. Remarkably, these conditions allow for imperfections both in the refocusing and in the measurement. Heisenberg scaling is thus achievable by choosing two additional passive and linear stages, whose roles are to conveniently scatter the input probe to all the $M$ modes, and then to refocus the photons. Despite the fact that the choice of these unitary stages can be in general $\varphi$-dependent, we show that it is always sufficient to adapt only one of the two stages, which we will thus call optimized, while the other stage can be chosen arbitrarily and independent of the parameter. Moreover we show that the optimized stage can be prepared
with a precision which is achievable using only classical resources, or by means of a preliminary classical estimation. 
This is also consistent with the result obtained in~\cite{GrTrFaNaTa}; namely, that a preliminary classical  estimation of $\varphi$ yields enough information to correctly prepare the optimized stage and thus to achieve Heisenberg scaling in the estimation protocol. Finally we show that the non-optimized stage affects the precision simply by a constant pre-factor. Using typicality and results of measure concentration in high-dimensional vector spaces, we show that distributing the unknown parameter among an high number of modes $M$ allows this pre-factor to typically take non-vanishing values.

The rest of the paper is organized as follows. In Section~\ref{sec:proposedSetup} we describe the proposed optical interferometer and the relative Fisher Information. In Section~\ref{sec:SufficientConditions} we use the Fisher Information to prove that the Heisenberg scaling can be achieved under suitable physical conditions; we then show how, even in the most general case, all the adaptivity can be confined within one of the auxiliary stages. In Section~\ref{sec:TypicalSensitivity} we discuss the typicality of our results for interferometers with a large number of channels. Finally, in Section~\ref{sec:conclusions} we draw some conclusions and discuss the outlook.

\section{The proposed setup}\label{sec:proposedSetup}

Let us consider a metrological scheme where the parameter to be estimated is encoded into an $M$-ports passive linear network described by the unitary $\hat{U}_\varphi$ acting on $M$ bosonic modes $\hat{a}_j$  ($j=1,\dots,M$) obeying the canonical commutation relations $[\hat{a}_j,\hat{a}_k^\dagger]=\delta_{jk}$, and $ [\hat{a}_j,\hat{a}_k]=[\hat{a}_j^\dagger,\hat{a}_k^\dagger]=0$. For a passive linear network, the action of $\hat{U}_\varphi$ on the annihilation operators is associated with an $M\times M$ unitary matrix via:
\begin{equation}
\hat{U}^\dag_\varphi\hat{a}_j\hat{U}_\varphi = \sum_{k=1}^M (U_\varphi)_{jk}\hat{a}_k.
\label{eq:UnitaryTransformation}
\end{equation}
The unitarity of the matrix $U_\varphi$ is strictly related to the conservation of the number of photons injected. By  definition, $U_\varphi$ is the matrix of the single-photon transition amplitudes, i.e. $\abs{(U_\varphi)_{jk}}^2$ is the probability that a single photon injected into the $k$-th input channel ends up in the $j$-th output channel due to the action of the network.

\begin{figure}[t!]
	\centering
	\includegraphics{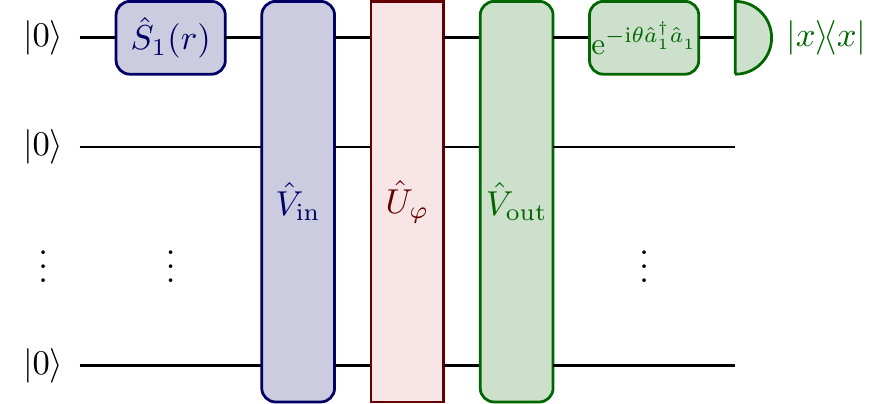}
	\caption{Block diagram of the investigated setup. A single-mode squeezed vacuum state with real squeezing parameter $r $ is injected into the first preparation stage $\hVi$, which inputs the linear network  $\hat{U}_\varphi$ encoding the parameter $\varphi$ to be estimated. After the network, there is a second stage $\hVo$ before the measurement. Finally, homodyne detection on the first  output port of $\hVo$ is performed, and the quadrature field $\hat{x}_\theta$ is measured. 
In order to reach Heisenberg scaling sensitivity in the estimation of $\varphi$, it suffices to optimize only one of the two auxiliary stages, in such a way that one of the conditions~\eqref{eq:ConditionOnVo} or~\eqref{eq:ConditionOnVi} holds.}
	\label{fig:Generic setup}
\end{figure}

We now propose an estimation scheme reaching Heisenberg scaling if suitable conditions are satisfied. As shown in \figurename{ \ref{fig:Generic setup}}, the preparation of the input probe consists in two steps: first, we inject a single-mode squeezed vacuum in the first port of $\hVi$.
Then, the unitary stage $\hVi$ is used to scatter the photons injected among all the modes.
The input state of the network $\hUphi$ in our protocol is therefore given by $\ket{\psi_0} = \hVi\hat{S}_1(r)\ket{\mathrm{vac}}$, 
where $\hat{S}_1(r)=\e^{\frac{r}{2}(\hat{a}_1^{2}-\hat{a}_1^{\dagger 2})}$ is a single-mode squeezing operator with squeezing parameter $r>0$, and $\ket{\mathrm{vac}}=\ket{0}^{\otimes M}$ is the $M$-mode vacuum state. The average number of photons injected in the apparatus is thus $N=\sinh^2 r$. The state $\hUphi\ket{\psi_0}$ at the output of the network $\hUphi$ undergoes the unitary $\hVo$ which refocuses all the photons into a single mode, namely the first one, where a homodyne measurement of the field quadrature $\hat{x}_\theta$ is performed. If the refocusing procedure is not perfect there will be some photons scattered into other channels with probability $1 - P_\varphi$, where 
\begin{equation}
\label{eq:ProbabilityTransition}
P_\varphi = 
\abs{(\Vo U_\varphi \Vi)_{11}}^2
\end{equation}
is defined by the the probability amplitude ${(\uphi)_{11} = (\Vo\Uphi\Vi)_{11}}$  for the transition from the first input to the first output port in the overall interferometer ${\uphi = \Vo\Uphi\Vi}$, with $\Vi$ and $\Vo$ being the single-photon unitary matrix representatives of $\hVi$ and $\hVo$ respectively, obtained analogously to \eqref{eq:UnitaryTransformation}.

The homodyne measurement is described by a Positive Operator Valued Measure (POVM) $\mathcal{M}=\{\hat{\Pi}_x\} $, whose elements are defined by
\begin{equation}
\hat{\Pi}_{x} =  \e^{\ii \theta \hat{a}_1^\dag\hat{a}_1}\ket{x}_{11}\!\bra{x}\e^{-\ii \theta \hat{a}_1^\dag\hat{a}_1}.
\label{eq:POVM}
\end{equation}

The probability of obtaining a value $x$  from a measurement of the quadrature $\hat{x}_\theta = \e^{\ii\theta\hat{a}^\dag_1\hat{a}_1}\hat{x}_1\e^{-\ii\theta\hat{a}^\dag_1\hat{a}_1}$ is then given by Born's rule
\begin{equation}
	p(x|\varphi)=\Tr(\hat{\Pi}_x \huphi \hat{S}_1(r)\ketbra{\mathrm{vac}}\hat{S}^\dag_1(r)\huphi^\dag),
	\label{eq:Born}
\end{equation} 
over the output state $\huphi\hat{S}_1(r)\ket{\mathrm{vac}}$ after the overall interferometric evolution $\huphi = \hVo\hUphi\hVi$, which yields (see Appendix \ref{app:ProbabilityDistribution})
\begin{equation}
p(x|\varphi) = \dfrac{1}{\sqrt{2\pi\Delta_\varphi}}\exp\left(-\frac{x^2}{2\Delta_\varphi}\right),
\label{eq:ProbabilityDistribution}
\end{equation}
where the variance of the Gaussian distribution
\begin{multline}
\Delta_\varphi = \dfrac{1}{2}\big(1+\abs{(\uphi)_{11}}^2(\cosh{2r }-1)+\\
+\Re[\e^{-2i\theta}(\uphi)_{11}^2]\sinh 2r \big),
\label{eq:HomodyneVariance}
\end{multline}
encodes the parameter $\varphi$ through the interferometric transition amplitude $(\uphi)_{11} = (\Vo\Uphi\Vi)_{11}$.

It is known from classical estimation theory that the maximum precision attainable when inferring the value of the unknown parameter $\varphi$, is given by the  so-called Cram{\'e}r-Rao bound~\cite{cramer1999mathematical,rao1992information}
\begin{equation}\label{eq:CRB}
\delta\varphi\geqslant \dfrac{1}{\sqrt{\nu F(\varphi)}}=\delta\varphi_\mathrm{min},
\end{equation}
where $\nu$ is the number of measurments performed, while 
$F(\varphi)$ is the Fisher Information (FI):
\begin{equation}
F(\varphi)=\int p(x|\varphi) \left(\frac{\partial \log p(x|\varphi)}{\partial \varphi}\right)^2 \ud x.
\label{eq:Fisher}
\end{equation}
The bound \eqref{eq:CRB} can be asymptotically saturated through post-processing Bayesian data analysis \cite{Pezze2008,olivares2009,berni2015}.

The Fisher Information related to a Gaussian probability distribution with variance $\Delta_\varphi$ can be evaluated by inserting \eqref{eq:ProbabilityDistribution} into \eqref{eq:Fisher}, and it reads
\begin{equation}
F(\varphi)=\dfrac{1}{2}\left(\dfrac{\partial_\varphi\Delta_\varphi}{\Delta_\varphi}\right)^2.
\label{eq:FisherGaussian}
\end{equation}
 
Plugging~\eqref{eq:HomodyneVariance} into~\eqref{eq:FisherGaussian}, it is possible to explicitly evaluate the FI (see Appendix \ref{Appendix:Fisher}) obtaining:

\begin{equation}
F(\varphi) = 2\left(\dfrac{(\partial_\varphi P_\varphi)f(N) -2P_\varphi(\partial_\varphi\gamma_\varphi) h(N)}{1+2P_\varphi f(N)}\right)^2,
\label{eq:Fisher0.1}
\end{equation}
where
\begin{equation}
f(N):=N\left(1+\cos[2(\gamma_\varphi-\theta)]\sqrt{1+1/N} \right),
\label{eq:fDefinition}
\end{equation}
\begin{equation}
h(N):=N\sin[2(\gamma_\varphi-\theta)]\sqrt{1+1/N},
\label{eq:hDefinition}
\end{equation}
with 
\begin{equation}
\gamma_\varphi=\arg\uphiff = \arg{(\Vo U_\varphi \Vi)_{11}}
\end{equation} 
being the accumulated phase through the interferometric evolution. 

\section{Heisenberg scaling}
\label{sec:SufficientConditions}

\begin{figure}[t!]
\includegraphics{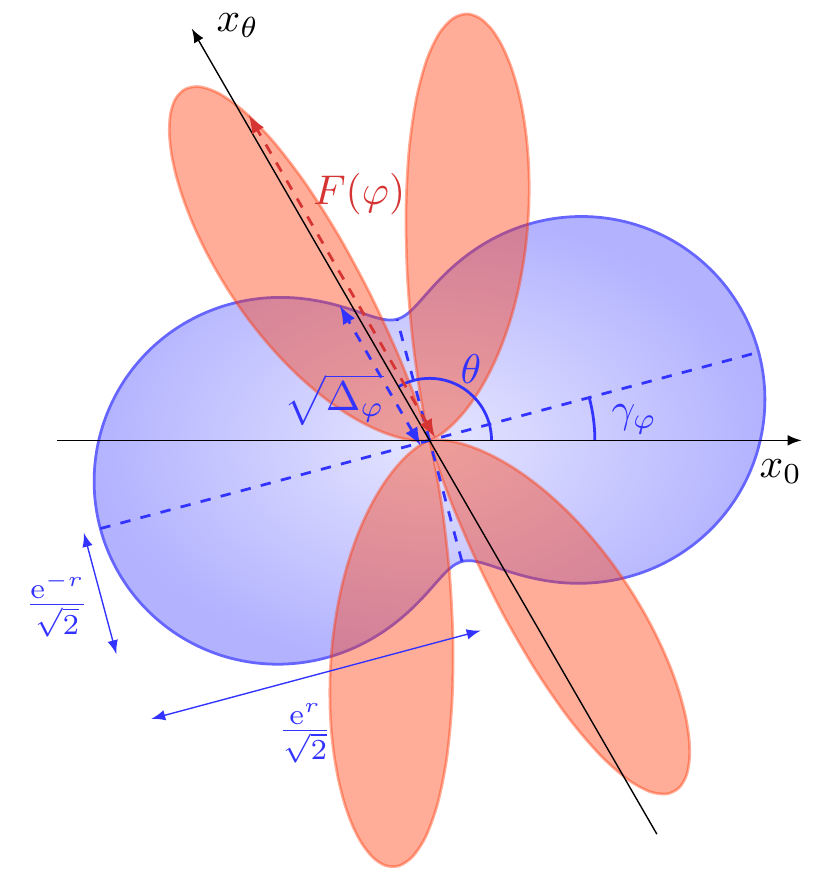}
\caption{Phase space representation of the squeezed vacuum state (with squeezing parameter $r\e^{2\ii\gamma_\varphi}$) at the first output channel of the whole setup shown in \figurename{ \ref{fig:Generic setup}} (blue), and of the Fisher information~\eqref{eq:Fisher0.1} (four red lobes). We have considered for simplicity the case where all the photons are refocused in the first output channel (when condition \eqref{eq:Condition1} reduces to $P_\varphi=1$). Given any axis at an angle $\theta$ with respect to the horizontal axis, corresponding to the $x_0=x$ quadrature, the distance between its intersections with the ellipse and the origin represents the standard deviation $\sqrt{\Delta_\varphi}$ of the quadrature $\hat{x}_\theta$. In other words, the blue graph is the polar plot of $\sqrt{\Delta_\varphi}$ shown in~\eqref{eq:HomodyneVariance} as a function of $\theta$. A polar plot of the Fisher Information is overlaid in red. The Fisher Information takes vanishing values if the minimum-variance quadratures are measured, namely for $\theta_{\mathrm{min}} = \gamma_\varphi \pm \pi/2$. This happens because the variance $\Delta_\varphi$ of the quadrature along $\theta_{\min}$ is locally insensitive to the variations of $\varphi$. Thus, one needs to get far enough from $\theta_{\min}$ to achieve a suitable high value of the Fisher Information. In particular, we have shown that for a large number $N$ of photons it is enough to move from $\theta_{\min}$ of an additional angle of the order $1/N$ as in \eqref{eq:Condition2} to reach the Heisenberg scaling in the measure of the parameter $\varphi$.
	}
\label{Fig:Theta&Gamma}
\end{figure}

We now demonstrate that Heisenberg scaling sensitivity can be achieved in the proposed metrological setup shown in \figurename{ \ref{fig:Generic setup}}, if two conditions are met. The first condition is the constraint that the average number of photons scattered into channels which are not measured is a finite quantity $\ell_{\varphi}$, independent of $N$, which translates into the condition
\begin{equation}
P_\varphi= 1 - \dfrac{\ell_{\varphi}}{N} +  O \left(\dfrac{1}{N^2}\right),\quad \ell_{\varphi}\geq0
\label{eq:Condition1}
\end{equation}
on the probability $P_\varphi$ in \eqref{eq:ProbabilityTransition}. Here, $\ell_\varphi$  depends in general on the linear network $\Uphi$ in which the parameter is embedded,  and on the auxiliary stages $\Vi$ and $\Vo$ in~\figurename\ref{fig:Generic setup}. In subsection~\ref{sec:OneSided} we will show how, for any given arbitrary $\Uphi$, it is possible to optimize only a single stage so that the probability distribution in \eqref{eq:ProbabilityTransition} can be expressed as \eqref{eq:Condition1}. The second condition relates the accumulated phase ${\gamma_\varphi = \arg\uphiff}$ through the whole setup and the phase $\theta=\theta_\varphi$ of the measured quadrature field $\hat{x}_\theta$, according to
\begin{equation}
\theta_\varphi = \gamma_\varphi \pm\dfrac{\pi}{2} + \dfrac{k_\varphi}{N} +  O \left(\dfrac{1}{N^2}\right),\quad \mathit{k}_\varphi\neq 0,
\label{eq:Condition2}
\end{equation}
where $k_\varphi$ can depend on $\varphi$, but is assumed to be independent of $N$. In practice, one can even fix  $k_\varphi$ to a constant value without using additional resources. 

A heuristic explanation behind this condition can be found in \figurename{ \ref{Fig:Theta&Gamma}}: in order to maximise the ratio in~\eqref{eq:FisherGaussian} while keeping constant $N=\sinh^2 r $, the choice of the quadrature  $\hat{x}_\theta$ to be measured is a  trade-off between two opposite behaviours. One consists in minimizing the variance $\Delta_\varphi$ in the denominator of~\eqref{eq:FisherGaussian}, while the other consists in maximizing the sensitivity of  the variance with respect to the variations of $\varphi$, namely choosing $\theta$ such that $\partial_\varphi\Delta_\varphi$ in the numerator is maximal. The former is met for $\theta$ as close as possible to $\gamma \pm\pi/2$, since $\hat{x}_{\gamma\pm\pi/2}$ are the squeezed quadratures after the rotation in phase space by the phase $\gamma_\varphi$  accumulated through the interferometer; the latter instead requires a choice of $\theta-\gamma$  far enough from the stationary points of the variance at $\gamma \pm \pi/2$, where $\Delta_\varphi$ and therefore the overall probability distribution $p(x|\varphi)$ are insensitive to variation in the parameter $\varphi$. Noticeably, the larger is $N$, and thus the squeezing parameter, the closer to the squeezed direction the quadrature field should be measured, as can be seen in~\eqref{eq:Condition2}

In order to prove the claim of HL scaling, we will evaluate the asymptotics of the Fisher information~\eqref{eq:Fisher0.1} as $N\rightarrow\infty$. Substituting the value $\theta=\theta_\varphi$ in~\eqref{eq:Condition2} into~\eqref{eq:fDefinition} and~\eqref{eq:hDefinition}, we get
\begin{equation}
f(N)= -\dfrac{1}{2}+\dfrac{2k_\varphi^2}{N}+\dfrac{1}{8N}+  O \left(\dfrac{1}{N^2}\right),
\label{eq:fAsymptotic}
\end{equation}
\begin{equation}
h(N) = 2k_\varphi\left(1+\dfrac{1}{2N}\right)+ O \left(\dfrac{1}{N^2}\right).
\label{eq:hAsymptotic}
\end{equation}
Hence, substituting~\eqref{eq:fAsymptotic} and~\eqref{eq:hAsymptotic} in~\eqref{eq:Fisher0.1}, and neglecting higher order terms, the asymptotic behavior of the Fisher information reads
\begin{equation}
F\left(\varphi\right)
	\sim 8\varrho(k_\varphi,\ell_{\varphi})(\partial_\varphi\gamma_\varphi)^2N^2,
	\label{eq:FisherWithVarrho}
\end{equation}
with
\begin{equation}
\varrho(k,\ell)=\left(\dfrac{8k}{1+16k^2+4\ell}\right)^2.
\label{eq:VarrhoDef}
\end{equation}
The quadratic scaling in the mean number photons $N$ in~\eqref{eq:FisherWithVarrho} finally proves that conditions~\eqref{eq:Condition1} and~\eqref{eq:Condition2} suffice to reach the Heisenberg scaling.

The asymptotics for the Fisher Information carries two pre-factors, $\varrho(k_\varphi,\ell_{\varphi})$ and $(\partial_\varphi \gamma_\varphi)^2$. We easily notice that the pre-factor $\varrho(k,\ell)$ vanishes only at $k=0$, and attains its maximum  at $\ell=0$, $k=\pm 1/4$:
\begin{equation}
\varrho(k,\ell)\leq\varrho(\pm 1/4,0)=1,
\end{equation}
so that, with this choice of the constants $k$ and $\ell$, the Fisher Information asymptotically reads
\begin{equation}\label{eq:FisherOptScaling}
F\left(\varphi\right)\Bigr|_{\substack{k=\pm\frac{1}{4}\\\ell=0}}\sim 8(\partial_\varphi\gamma_\varphi)^2N^2
\end{equation}
Moreover, $\varrho(k,\ell)$ is a decreasing function of $\ell$ independent of $k$, so that $\ell=0$ is always the best case, meaning that the less photons that are scattered in different channels, the higher the sensitivity in the estimation. Instead, for a fixed arbitrary positive value of $\ell$, the maximum of $\varrho(k,\ell)$ is reached for $k=\pm\sqrt{4\ell+1}/4$. 

\subsection{One-sided adaptivity}
\label{sec:OneSided}
Since $P_\varphi=\abs{(\Vo U_\varphi \Vi)_{11}}^2$, condition~\eqref{eq:Condition1} may appear to require a simultaneous optimization of the input $\Vi$ and the output $\Vo$ in a parameter-dependent way.
This two-sided adaptation can be quite difficult to realize in practice.

However, we are going to show that in fact conditions~\eqref{eq:Condition1} and~\eqref{eq:Condition2} can always be satisfied with just a one-sided parameter-dependent adaptation, which can be performed either at the input or at the output of the network equivalently. And remarkably, this adaptation can be accomplished by performing a preliminary classical, shot-noise limited, estimation of $\varphi$. 

In particular, one can choose to adaptively optimize only $\Vo$ and fix $\Vi$ to an arbitrary parameter-independent unitary stage: in this case, one can set the parameter-dependent condition
\begin{equation}
(\Vo)_{1i}=(\Vi^\dag U_\varphi^\dag)_{1i}+ O \left(\frac{1}{\sqrt{N}}\right).
\label{eq:ConditionOnVo}
\end{equation}
Alternatively, it is possible to adaptively optimize only $\Vi$ with the condition
\begin{equation}
(\Vi)_{i1} = (U_\varphi^\dag \Vo^\dag)_{i1}+ O \left(\frac{1}{\sqrt{N}}\right),
\label{eq:ConditionOnVi}
\end{equation}
while 
$\Vo$ can be  arbitrarily chosen. 

Remarkably both equations~\eqref{eq:ConditionOnVo} and~\eqref{eq:ConditionOnVi} imply that an error of the order of $ O \left(\frac{1}{\sqrt{N}}\right)$ is allowed to prepare the optimized stage to reach condition~\eqref{eq:Condition1}. To show that both equations~\eqref{eq:ConditionOnVo} and~\eqref{eq:ConditionOnVi} satisfy condition \eqref{eq:Condition1}, we notice that $P_\varphi$ can be expressed as a transition probability
\begin{equation}
	P_\varphi=\abs{(\Vo U_\varphi \Vi)_{11}}^2=\abs{\braket{v_{\mathrm{out}}}{v_{\mathrm{in}}}}^2,
\end{equation}
between {the two normalized vectors $\ket{v_{\mathrm{in}}} = U_\varphi \Vi\ket{e_1}$ and $\ket{v_{\mathrm{out}}} = \Vo^\dagger\ket{e_1}$,
with $\ket{e_1}=(1,0,\dots,0)^T$. Then, equation~\eqref{eq:ConditionOnVo} translates into
\begin{equation}\label{eq:epsilon}
	\ket{v_{\mathrm{out}}} =  \ket{v_{\mathrm{in}}} + \ket{\delta v} =  \e^{i\varepsilon H} \ket{v_{\mathrm{in}}}, \quad \varepsilon= O \left(\frac{1}{\sqrt{N}}\right),
\end{equation} 
with some $H=H^\dagger$, by unitarity.
Therefore, we can see that
\begin{align}\notag
	P&=\abs{\braket{v_{\mathrm{out}}}{v_{\mathrm{in}}}}^2=\abs{\braket{v_{\mathrm{in}}}{\e^{-i\varepsilon H}v_{\mathrm{in}}}}^2\\
	\notag
	&= \abs{1-i \varepsilon \braket{v_{\mathrm{in}}}{H v_{\mathrm{in}}} - O(\varepsilon^2) }^2 = 1-O(\varepsilon^2)\\
	&=1- O \left(\frac{1}{N}\right)
\end{align}

We can notice that in both equations~\eqref{eq:ConditionOnVo} and~\eqref{eq:ConditionOnVi}, no assumption on the non-optimized stage is made, so that its choice is completely arbitrary. This freedom affects the precision of the estimation of $\varphi$ through the $N$-independent pre-factor $(\partial_\varphi\gamma_\varphi)^2$ which appears in the Fisher Information~\eqref{eq:FisherWithVarrho}. At this point, one may argue that this pre-factor may be vanishing if a poor choice for the non-adapted unitary is made. Remarkably, in the next Section we will show that the pre-factor is typically non-vanishing for random choices of the non-adapted stage and suitably well-behaved given linear networks $\hUphi$.

\section{Typical sensitivity}\label{sec:TypicalSensitivity}
In this section we will address in more detail the study of the pre-factor $(\partial_\varphi\gamma_\varphi)^2$ in the Fisher information \eqref{eq:FisherWithVarrho}, clarifying under what circumstances it can be safely considered non-vanishing and characterizing its magnitude for random choices of the non-optimized stage. First of all, we can link $(\partial_\varphi\gamma_\varphi)^2$ to the derivative of the matrix element $(u_\varphi)_{11}=(\Vo U_\varphi \Vi)_{11}=\sqrt{P_\varphi}\e^{\ii \gamma_\varphi}$:
\begin{align}\notag
\abs{(\partial_\varphi u_\varphi)_{11}}^2&=\abs{\left(\partial_\varphi \sqrt{P_\varphi}+i(\partial_\varphi \gamma_\varphi)\sqrt{P_\varphi}\right)\e^{\ii\gamma_\varphi}}^2\\
&=(\partial_\varphi \sqrt{P_\varphi})^2+(\partial_\varphi \gamma_\varphi)^2 P_\varphi.
\label{eq:du11}
\end{align}
If condition~\eqref{eq:Condition1} is satisfied, equation~\eqref{eq:du11} simplifies to
\begin{equation}\label{eq:du11dga}
(\partial_\varphi \gamma_\varphi)^2 = \abs{(\partial_\varphi u_\varphi)_{11}}^2+  O \left(\frac{1}{N}\right),
\end{equation}
so that the two quantities are equal up to order $1/N$.

Now, if the adaptation is performed in the output, i.e. we choose an arbitrary $\Vi$ and adapt $\Vo$ according to equation~\eqref{eq:ConditionOnVo}, we see that
\begin{equation}
	\abs{(\partial_\varphi u_\varphi)_{11}}^2=(\Vi^\dagger G_\varphi \Vi)_{11}^2+ O \left(\frac{1}{N}\right),
\label{eq:VinArbit}
\end{equation}
where the Hermitian operator
\begin{equation}
	G_\varphi:=\ii U_\varphi^\dagger \frac{\partial_\varphi U_\varphi}{\partial \varphi}
\end{equation}
is the ($\varphi$-dependent) generator of $U_\varphi$.
If, on the other hand, condition~\eqref{eq:Condition1} is realized through an adaptation on the input while taking an arbitrary $\Vo$, then equation~\eqref{eq:ConditionOnVi} implies that
\begin{equation}
\abs{(\partial_\varphi u_\varphi)_{11}}^2=(\Vo U_\varphi G_\varphi U_\varphi^\dagger \Vo^\dagger)_{11}^2+ O \left(\frac{1}{N}\right).
\label{eq:VoutArbit}
\end{equation}

Using equations~\eqref{eq:du11dga}-\eqref{eq:VoutArbit}, we can finally rewrite the asymptotic expression of the Fisher information~\eqref{eq:FisherWithVarrho} as
\begin{equation}\label{eq:FisherPrefactorf}
	F(\varphi)\sim \varrho(k_\varphi,\ell_{\varphi}) f(U,G_\varphi)N^2,
\end{equation}
as $N\to\infty$, where
\begin{equation}\label{eq:prefactor}
f(U,G_\varphi)=(U^\dagger G_\varphi U)_{11}^2,
\end{equation}
with $U=\Vi$ if the optimization is performed on the output, while $U=U_\varphi^\dagger \Vo^\dagger$ if the optimization is carried out on the input.  We emphasize that the pre-factor $\f$ is completely independent of the choice of the optimized stage.

The maximization of the prefactor~\eqref{eq:prefactor} can be realized, for example, if $U=V_\varphi$ is some unitary diagonalizing $G_\varphi$ i.e. satisfying equation $V_\varphi^\dagger G_\varphi V_\varphi=D_\varphi$ with $D_\varphi=\diag(g_1,g_2,\dots,g_M)$ being the diagonal matrix of the eigenvalues of $G_\varphi$, ordered in such a way that $\abs{g_1}=\Norm{G_\varphi}$ is the maximum eigenvalue in absolute value~\cite{Matsubara_2019}.  Actually, it is not necessary to take a diagonalizing unitary to maximize~\eqref{eq:prefactor}, since only the first column of $U$ enters in the definition of $f(U,G_\varphi)$; hence, to maximize $f(U,G_\varphi)$ it is sufficient to require this column to be the eigenvector of $G_\varphi$ corresponding to the maximum eigenvalue $\Norm{G_\varphi}$. 
However, even that requirement would necessitate the complete knowledge of $G_\varphi$, which in general depends on the unknown parameter $\varphi$.
Therefore,  it is more relevant to consider arbitrary choices of the non-adapted network (the unitary $U$) independently of $\varphi$ in order to determine the practical advantages of the obtained Heisenberg scaling precision for finite values of $N$ and only one (classically) adapted stage. 

For this reason, we will perform now a statistical analysis on the typical values which can be assumed by the prefactor $f(U,G_\varphi)$ for random choices of the unitary $U$.
Assuming no prior knowledge of the unitary $U$, we sample it from the unitary group $\mathrm{U}(M)$ according to the  unbiased uniform distribution probability, i.e. the unitarily invariant Haar measure $\mathcal{P}$.

For a random unitary  $U$, sampled according to this distribution, the average value of the prefactor $f(U,G_\varphi)$ can be computed using techniques from random matrix theory (see Appendix~\ref{app:averagesU}):
\begin{equation}
\mathbb{E}[f(U,G_\varphi)] 
=\frac{\Tr(G_\varphi^2)+\Tr(G_\varphi)^2}{M(M+1)},
\label{eq:avef}
\end{equation}
where $\E[\cdot]$ denotes the expectation value over $\mathrm{U}(M)$  with respect to the Haar measure. 

In the trivial case of a generator proportional to the identity, $G_\varphi=\Norm{G_\varphi}\mathbbm{1}$, which corresponds to the case of a network ${U_\varphi=\e^{\ii\Norm{G_\varphi}}\mathbbm{1}}$ acting as a $\varphi$-dependent global phase shifter, we have $\Tr(G_\varphi)^2=M^2\Norm{G_\varphi}^2$ and $\Tr(G_\varphi^2)=M\Norm{G_\varphi}^2$, so that the average value of the pre-factor  equals the maximum one, $f_\mathrm{max}=\Norm{G_\varphi}^2$, in accordance to the fact that in this particular case \textit{every} unitary in $\mathrm{U}(M)$ diagonalizes $G_\varphi$. 

In general, we are interested in determining the conditions which make this average value in \eqref{eq:avef} as large as possible. First of all, we can note directly from expression~\eqref{eq:avef} that eigenvalues of opposite signs can have a detrimental effect on this average, since they lower the value of $\Tr(G_\varphi)$. In general, we can find a lower bound on the average value~\eqref{eq:avef} using Jensen's inequality $\E[X^2]\geqslant \E[X]^2$ to obtain
\begin{equation}\label{eq:lowbound}
\E[f(U,G_\varphi)]\geqslant \E[(U^\dagger G_\varphi U)_{11}]^2=\left[ \frac{\Tr (G_\varphi)}{M} \right]^2,
\end{equation}
where again, the average $\E[(U^\dagger G_\varphi U)_{11}]$ has been computed using standard techniques (see Appendix \ref{app:averagesU}).

\begin{figure}[]
	\centering
	\includegraphics[width=.49\textwidth]{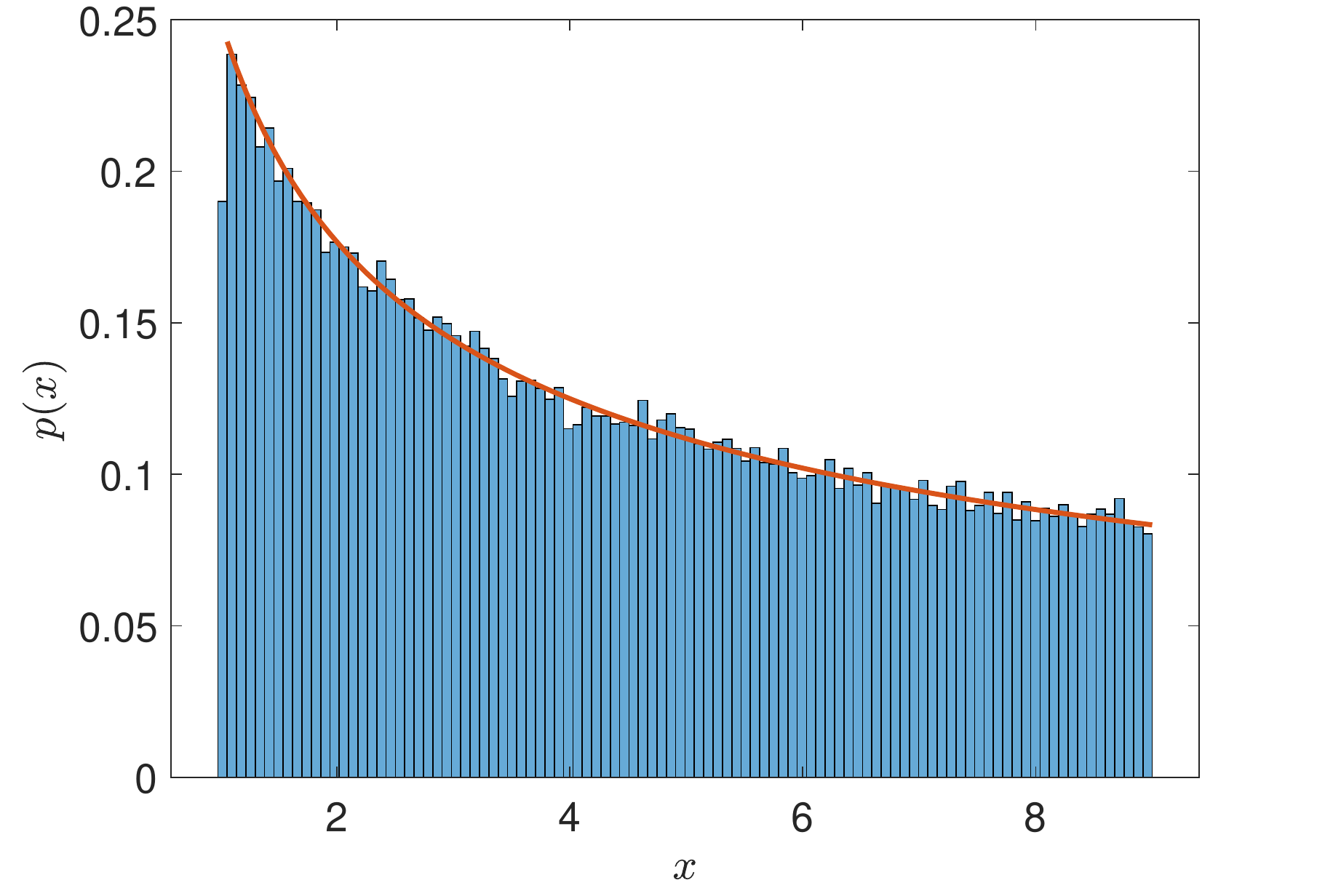}\\
	\includegraphics[width=.49\textwidth]{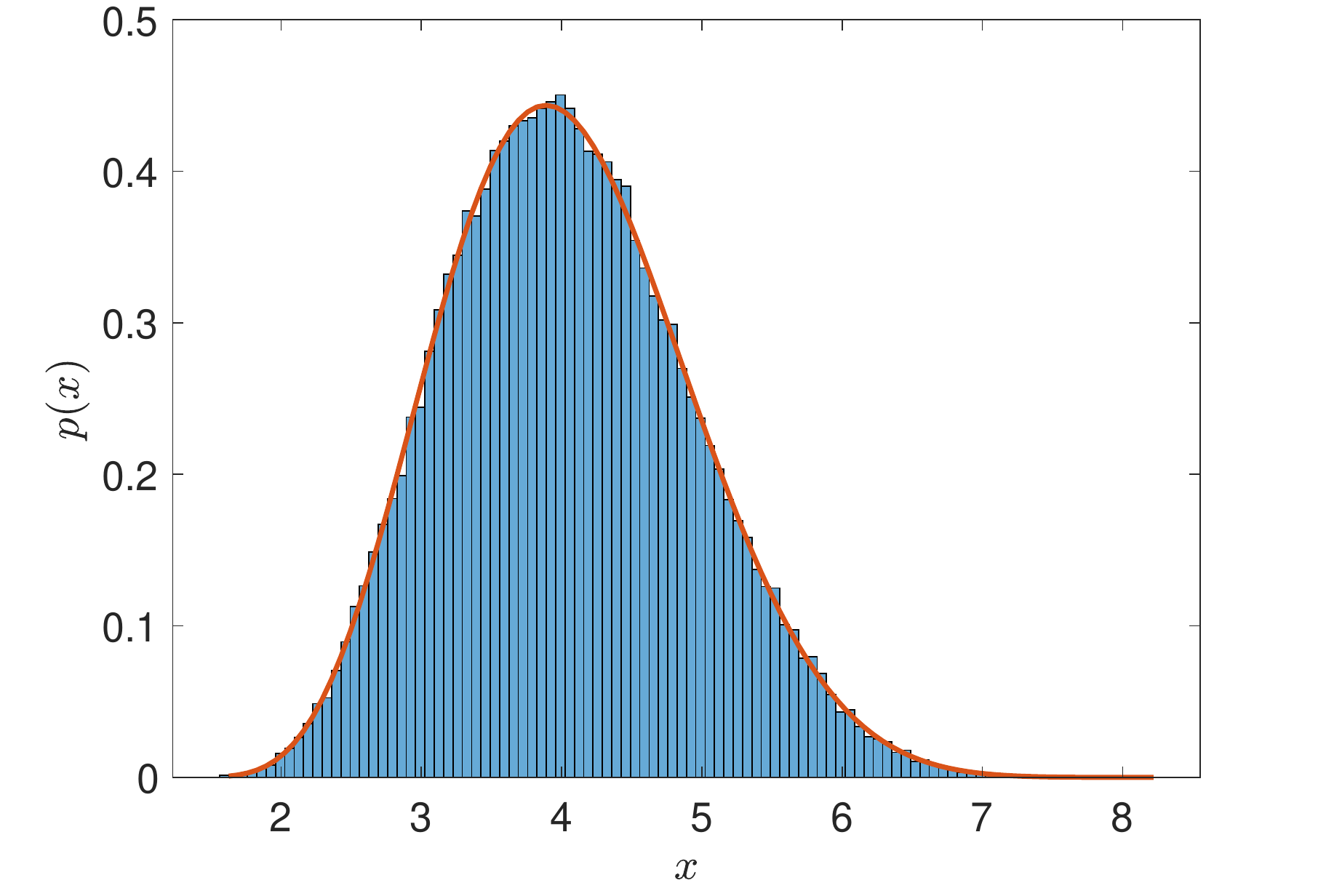}\\
	\includegraphics[width=.49\textwidth]{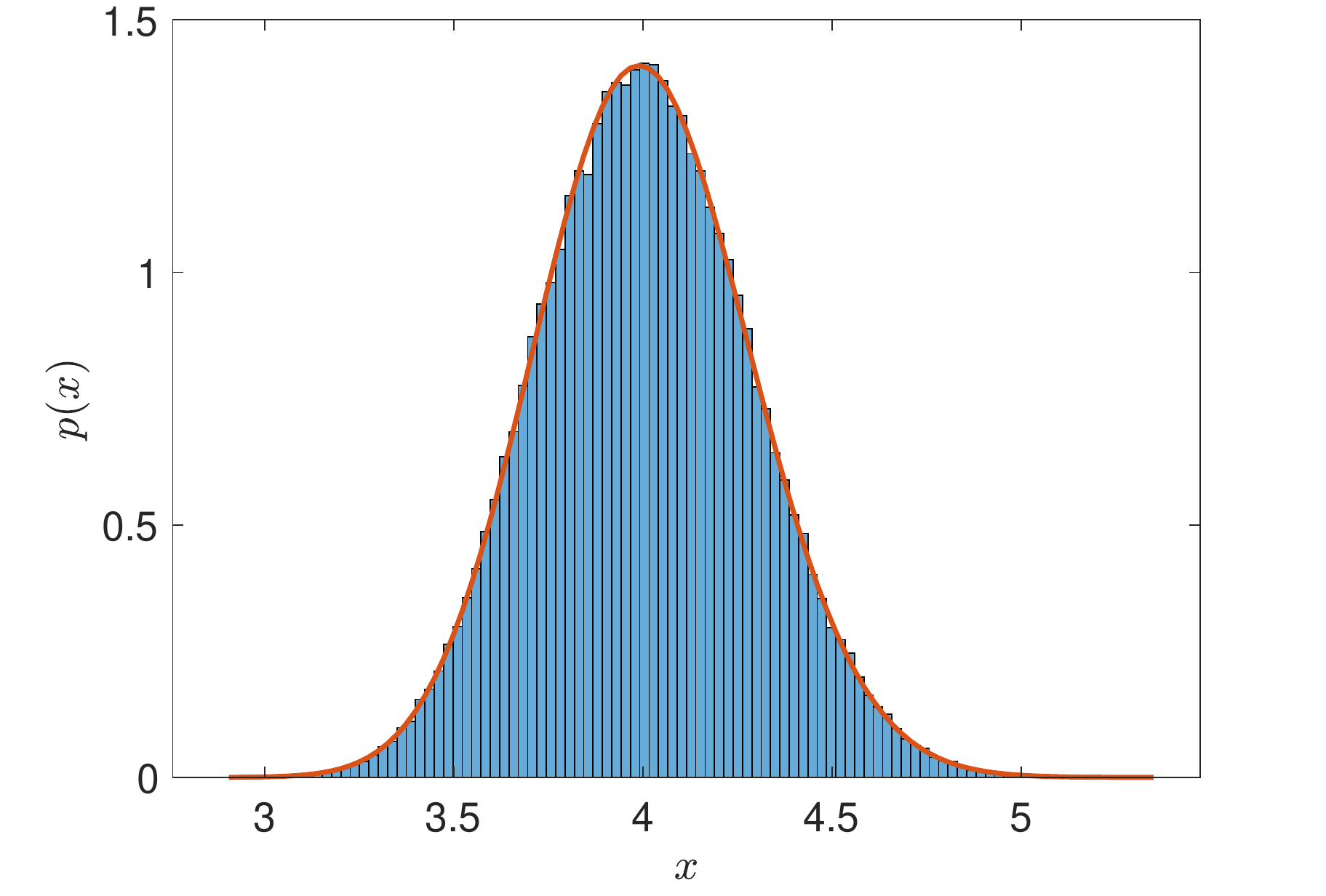}
	\caption{Histograms of the random variable $f(U,G_\varphi)=(U^\dagger G_\varphi U)_{11}^2$ numerically obtained with $10^5$ samplings of $U$ from the unitary group with Haar measure and choosing $G_\varphi$ as a diagonal matrix with half entries $3$s and half entries made of $1$s, with $M=2$ (top), $M=20$ (middle) and $M=200$ (bottom). The normalization is chosen in such a way that the total area under the curve equals $1$. For these particular cases an explicit analytic expression of the distribution is achievable, and is given by  the orange curves. The derivation is shown in Appendix \ref{app:explicitDist}.}
	\label{fig:FIdist}
\end{figure}

Notice that the right-hand side of this inequality is nothing but the square of the average between $G_\varphi$'s eigenvalues. Hence, if we have some degree of control on the eigenvalues, we can achieve a result which is a certain fraction $\alpha$ of the maximum value $f_\mathrm{max}$ if the average of $G_\varphi$'s eigenvalues is at least a fraction $\sqrt{\alpha}$ of the maximum eigenvalue, namely:
\begin{equation}\label{eq:lowerbound}
\left[ \frac{\Tr (G_\varphi)}{M} \right]^2\geqslant \alpha \Norm{G_\varphi}^2 \quad\Rightarrow \quad \E[f(U,G_\varphi)]\geqslant \alpha f_\mathrm{max}.
\end{equation}

However, this may be not sufficient for our purposes, since the average of a random variable alone does not determine its \textit{typical} behaviour: a paradigmatic elementary example is that of a real random variable taking only the values $0$ or $1$ with equal probabilities, thus having an average of $1/2$ even if it \textit{never} takes values close to $1/2$. 

We will show now that this is not the case for the pre-factor $f(U,G_\varphi)$, thanks to the fact that it is a sufficiently well behaved function with respect to the random unitary $U$. In fact, by using  results on concentration of measure in high-dimensional probability spaces, we prove in Appendix~\ref{app:concentration} that for an network with a large number $M$ of ports the pre-factor $f(U,G_\varphi)$ becomes \textit{typical}, meaning that it becomes almost constant with respect to random choices of $U\in\mathrm{U}(M)$ (according to the unitarily invariant measure), hence concentrating around its average value~\eqref{eq:avef},  bounded below by~\eqref{eq:lowerbound}.

In formulas, we have that
\begin{equation}
\mathcal{P}(\abs{f-\mathbb{E}[f]}\geqslant \varepsilon)\leqslant 2 \exp(-\frac{A M}{\Norm{G_\varphi}^4}\varepsilon^2).
\label{eq:concentration}
\end{equation}
where $A=(72\pi^3)^{-1}$. This result tells us that for large interferometers it is extremely unlikely to obtain a pre-factor sensibly different from its average, since for large values of $M$ the probability of $f(U,G_\varphi)$ being different from its average is exponentially suppressed.

This result can be also seen from the exact distribution of the pre-factor computed for some particular cases, which is shown in \figurename{ \ref{fig:FIdist}}. It can be seen from these figures that as $M$ is increased the distribution of the pre-factor concentrates around its average: in particular, for the chosen configuration, a value of $M=20$ is already sufficient to get this concentration. Thus, for any well-behaved linear network $\hUphi$ such that the expectation value in \eqref{eq:avef} is far enough from zero, any random choice of the non-adapted stage in the proposed interferometric setup typically yields an Heisenberg-scaling precision for the estimation of $\varphi$ if the number $M$ of interferometric channels is large enough.

\section{Conclusions}\label{sec:conclusions}
We demonstrated by using a simple metrological technique the typicality of  Heisenberg scaling precision for the estimation of a generic parameter $\varphi$ encoded into an arbitrary $M$-mode network. Our scheme can be applied regardless of the nature of the parameter, which can even  be distributed among several components of the network. In particular, the proposed scheme  makes use of a single-mode squeezed state as a probe, scattered throughout all the modes by means of an auxiliary passive linear stage. Once the information on the parameter is gathered by the probe, this gets refocused on a single output channel by a second auxiliary stage, and then detected with homodyne measurement. The analysis of the Fisher information associated with such scheme reveals that, if a constant average number of photons (not scaling with the total number of photons injected) is scattered into channels different from the one measured, due to an imperfect refocusing procedure, the Heisenberg limit can be asymptotically reached, provided that the homodyne detection is performed with a sufficient resolution. For a distributed parameter, the refocusing is generally parameter-dependent, implying some sort of adaptive procedure in order to correctly refocus the probe. However, we have shown that all the dependence on the parameter can be entirely bounded to only one of the two auxiliary stages, while the other only affects the estimation through a multiplicative pre-factor. Moreover, we have also discussed how all the information on the parameter needed to sufficiently refocus the probe can be obtained with a classical shot-noise precision, meaning that the number of resources required to adaptively optimize the auxiliary stages is not detrimental for the Heisenberg scaling precision. 
Finally, we have shown that, for a large number of modes, Heisenberg scaling is typically obtained by an arbitrary non-adapted stage, with an overwhelming probability, i.e.\ an exponentially suppressed probability of failure.

\section*{Acknowledgements}
This work was supported by the Office of
Naval Research Global (N62909-18-1-2153). PF and GG are partially supported by Istituto Nazionale di Fisica Nucleare (INFN) through the project ``QUANTUM'', and by the Italian National Group of Mathematical Physics (GNFM-INdAM).

GG and DT contributed equally to the drafting of this work.

\appendix

\section{Derivation of the probability distribution in~\eqref{eq:ProbabilityDistribution}}\label{app:ProbabilityDistribution}

In order to evaluate the probability distribution associated to the measure of the quadrature field $x_\theta$, it is natural to proceed with the phase-space formalism for Gaussian quantum optics. For a deep and systematic overview in this topic, several works and reviews can be found in literature~\cite{Scully1997,Schleich,Weedbrook2012}. In this appendix, we will briefly introduce only the concepts and tools needed to obtain the expression~\eqref{eq:ProbabilityDistribution} for the probability distribution $p(x|\varphi)$ of measuring the quadrature value $x$ at the output of the proposed interferometer in \figurename~\ref{fig:Generic setup}. We recall first the definition \eqref{eq:Born}
\begin{equation}
p(x|\varphi)=\bra{\mathrm{vac}}\hat{S}^\dag_1(r )\hat{u}^\dag_\varphi \e^{\ii \theta \hat{a}_1^\dag\hat{a}_1}|x\rangle_{11}\langle x|\e^{-\ii \theta \hat{a}_1^\dag\hat{a}_1}\hat{u}_\varphi\hat{S}_1(r )\ket{\mathrm{vac}},
\label{eq:PDFAppendix}
\end{equation}
where $\huphi$ describes the overall interferometric evolution of the single mode squeezed state $\hat{S}_1(r)\ket{\mathrm{vac}}$. To evaluate~\eqref{eq:PDFAppendix}, it is useful to firstly recover its Fourier transform,
\begin{align}
\chi(\xi|\varphi)&=\int \d x\, p(x|\varphi)\e^{-\ii x\xi}\notag\\
&=\bra{\mathrm{vac}}\hat{S}^\dag_1(r )\hat{u}^\dag_\varphi \e^{\ii \theta \hat{a}_1^\dag\hat{a}_1}\e^{-\ii \hat{x}_1\xi}\e^{-\ii \theta \hat{a}_1^\dag\hat{a}_1}\hat{u}_\varphi\hat{S}_1(r )\ket{\mathrm{vac}}.
\label{eq:FourierTransform}
\end{align}
It is possible to write this characteristic function in a more canonical way. Indeed, first notice that we can write (the derivation is given in Appendix~\ref{app:Displacement} for completeness)
\begin{equation}\label{eq:Displacement}
\e^{\ii \theta \hat{a}_1^\dag\hat{a}_1}\e^{-\ii \hat{x}_1\xi}\e^{-\ii \theta \hat{a}_1^\dag\hat{a}_1}=
\hat{D}(\vec{\xi}_\theta),
\end{equation}
where 
\begin{equation}
\hat{D}(\vec{\xi}_\theta)=\e^{-\ii \vec{\xi}_\theta\cdot\vec{z}}
\end{equation} 
is the \textit{displacement operator},  with
\begin{equation}
\label{eq:DefinitionXiAndZ}
\vec{\xi}_\theta=
\begin{pmatrix}
\xi\cos\theta\\
\xi\sin\theta
\end{pmatrix},\quad
\vec{\hat{z}}=
\begin{pmatrix}
\hat{x}_1\\
\hat{p}_1
\end{pmatrix}.
\end{equation}
Then, using equation~\eqref{eq:Displacement} we can write the characteristic function~\eqref{eq:FourierTransform} as
\begin{equation}
	\chi(\xi|\varphi)=\Tr[\hat{D}(\vec{\xi}_\theta)\hat{u}_\varphi\hat{S}_1(r )\ketbra{\mathrm{vac}}\hat{S}^\dag_1(r )\hat{u}^\dag_\varphi],
	\label{eq:CharacteristicFunctionDisplacement}
\end{equation}
Due to the Gaussian nature of the squeezed vacuum state and the linearity of the interferometric setup, the characteristic function~\eqref{eq:CharacteristicFunctionDisplacement} is a Gaussian bivariate function centred in zero, of the form~\cite{Weedbrook2012}
\begin{equation}
\chi(\xi|\varphi)=\e^{-\frac{1}{2}\vec{\xi}^\mathrm{T}_\theta\mathbb{\sigma}_\varphi\vec{\xi}^\mathrm{\phantom{T}}_\theta},
\label{eq:CharacteristicFunction}
\end{equation}
where $\sigma_\varphi$ is the $2\times 2$ covariance matrix of the whole interferometer output state $\huphi\hat{S}_1(r)\ket{\mathrm{vac}}$, reduced to the first mode. 

In order to evaluate this matrix, we firstly recover the covariance matrix $\Gamma_0$ of the input state $\hat{S}_1(r )\ket{\mathrm{vac}}$, which reads
\begin{equation}
\Gamma_0 =\dfrac{1}{2}
\begin{pmatrix}
\e^{2\mathcal{R}} & 0\\
0 & \e^{-2\mathcal{R}}
\end{pmatrix},
\end{equation}
where $\mathcal{R}$ is the $M\times M$ diagonal matrix with a single non-zero entry $\mathcal{R}_{11}\equiv r $. After the action of the interferometer, the covariance matrix transforms into
\begin{equation}
\label{eq:AppGammaPhi}
\Gamma_\varphi = R_\varphi \Gamma_0 R_\varphi^T,
\end{equation}
where $R_\varphi$ is the orthogonal and symplectic matrix associated with the interferometer unitary matrix $\uphi$
\begin{equation}
R_\varphi = W^\dag 
\begin{pmatrix}
u_\varphi & 0\\
0 & u^*_\varphi
\end{pmatrix}
W,
\end{equation}
and
\begin{equation}
W = \dfrac{1}{\sqrt{2}}
\begin{pmatrix}
\mathbbm{1} & i\mathbbm{1}\\
\mathbbm{1} & -i\mathbbm{1}
\end{pmatrix},
\end{equation}
where $\mathbbm{1}$ is the $M\times M$ identity matrix. $R_\varphi$ can be easily evaluated to be
\begin{equation}
R_\varphi =
\begin{pmatrix}
\re{u_\varphi} & -\im{u_\varphi}\\
\im{u_\varphi} & \re{u_\varphi}
\end{pmatrix},
\end{equation}
so that $\Gamma_\varphi$ in~\eqref{eq:AppGammaPhi} reads
\begin{equation}
\Gamma_\varphi =
\begin{pmatrix}
\Delta X^2_\varphi & \Delta XP_\varphi\\
(\Delta XP_\varphi)^T & \Delta P^2_\varphi
\end{pmatrix}.
\end{equation}
where we have defined the $M\times M$ matrices
\begingroup
\allowdisplaybreaks
\begin{align}
\Delta X^2_\varphi &\equiv  \dfrac{1}{2}\left[\re{\uphi}\e^{2\mathcal{R}}\re{\uphi^\dag} - \im{\uphi}\e^{-2\mathcal{R}}\im{\uphi^\dag}\right]\notag\\
&=\dfrac{1}{2}\left[\Re[\uphi\cosh(2\mathcal{R})\uphi^\dag]+\Re[\uphi\sinh(2\mathcal{R})\uphi^\mathrm{T}]\right],\\
\Delta P^2_\varphi &\equiv \dfrac{1}{2}\left[ -\im{\uphi}\e^{2\mathcal{R}}\im{\uphi^\dag} + \re{\uphi}\e^{-2\mathcal{R}}\re{\uphi^\dag}\right]\notag\\
&= \dfrac{1}{2}\left[\Re[\uphi\cosh(2\mathcal{R})\uphi^\dag]-\Re[\uphi\sinh(2\mathcal{R})\uphi^\mathrm{T}]\right],\\
\Delta XP_\varphi &\equiv \dfrac{1}{2}\left[ -\re{\uphi}\e^{2\mathcal{R}}\im{\uphi^\dag} - \im{\uphi}\e^{-2\mathcal{R}}\re{\uphi^\dag}\right]\notag\\
&=\dfrac{1}{2}\left[-\Im[\uphi\cosh(2\mathcal{R})\uphi^\dag]+\Im[\uphi\sinh(2\mathcal{R})\uphi^\mathrm{T}]\right].
\end{align}
\endgroup
In the second lines of each of the previous expression, we have exploited the fact that $\mathcal{R}$ is real. We are interested to evaluate $\sigma_\varphi$, the covariance matrix reduced to the first mode, which we can now readily write
\begin{equation}
\sigma_\varphi =
\begin{pmatrix}
\left(\Delta X^2_\varphi\right)_{11} & \left(\Delta XP_\varphi\right)_{11}\\
\left(\Delta XP_\varphi\right)_{11} & \left(\Delta P^2_\varphi\right)_{11}
\end{pmatrix},
\end{equation}
and insert in~\eqref{eq:CharacteristicFunction}. Our final step is to invert the Fourier transform to finally get the expression of the probability distribution $p(x|\varphi)$ given in~\eqref{eq:ProbabilityDistribution}. In order to do that, we introduce the $2\times2$ orthogonal matrix
\begin{equation}
O_\theta =
\begin{pmatrix}
\cos\theta & -\sin\theta\\
\sin\theta & \cos\theta
\end{pmatrix},
\end{equation}
such that $\vec{\xi}_\theta = O_\theta\vec{\xi}_0$, with $\vec{\xi}_0 = (\xi,0)^\mathrm{T}$. Then, the characteristic function~\eqref{eq:CharacteristicFunction} can be written in a more convenient way, namely
\begin{equation}
\chi(\xi|\varphi) = \e^{-\frac{1}{2}\vec{\xi}^\mathrm{T}_0 O_\theta^\mathrm{T}\mathbb{\sigma}_\varphi O_\theta\vec{\xi}^\mathrm{\phantom{T}}_0}=\e^{-\frac{1}{2}\left(O_\theta^\mathrm{T}\mathbb{\sigma}_\varphi O_\theta\right)_{11}\xi^2}.
\end{equation}
Exploiting, by the definition of $\mathcal{R}$, the identities
\begin{align}
\left(\uphi\cosh(2\mathcal{R})\uphi^\dag\right)_{11}&= \cosh(2r )\abs{(\uphi)_{11}}^2+ \sum_{i=2}^M \abs{(\uphi)_{1i}}^2 \notag\\
= \cosh(2r )&\abs{(\uphi)_{11}}^2+\left(1-\abs{(\uphi)_{11}}^2\right),\\
\left(\uphi\sinh(2\mathcal{R})\uphi^T\right)_{11}&=\sinh(2r )(\uphi)_{11}^2,
\end{align}
and some elementary trigonometry, the term $\left(O_\theta^\mathrm{T}\mathbb{\sigma}_\varphi O_\theta\right)_{11}$ can be further manipulated to match the expression of $\Delta_\varphi$ given in~\eqref{eq:HomodyneVariance}. In fact
\begin{align}
\big(&O_\theta^\mathrm{T}\mathbb{\sigma}_\varphi O_\theta\big)_{11} = \sum_{i,j=1,2}(O_\theta)_{i1}(O_\theta)_{j1}(\sigma_\varphi)_{ij}\notag\\
&=\cos^2\theta \left(\Delta X^2_\varphi\right)_{11} + \sin^2\theta \left(\Delta P^2_\varphi\right)_{11}\notag+\\
&\quad+2\cos\theta\sin\theta\left(\Delta XP_\varphi\right)_{11}\notag\\
&=\dfrac{1}{2}\cos^2\theta\left(\cosh(2r )\abs{(\uphi)_{11}}^2+\sinh(2r )\Re((\uphi)_{11}^2)\right)+\notag\\
&\quad+\dfrac{1}{2}\sin^2\theta\left(\cosh(2r )\abs{(\uphi)_{11}}^2-\sinh(2r )\Re((\uphi)_{11}^2)\right)+\notag\\
&\quad+\dfrac{1}{2}\left(\cos^2\theta+\sin^2\theta\right)\left(1-\abs{(\uphi)_{11}}^2\right)\notag+\\
&\quad+\cos\theta\sin\theta\sinh(2r )\Im[(\uphi)_{11}^2]\notag\\
&=\dfrac{1}{2}\Big(1+\abs{(\uphi)_{11}}^2(\cosh(2r )-1)\notag+\\
&\quad+\sinh(2r )\left(\cos(2\theta)\Re((\uphi)_{11}^2)+\sin(2\theta)\Im[(\uphi)_{11}^2]\right)\Big)\notag\\
&=\dfrac{1}{2}\Big(1+\abs{(\uphi)_{11}}^2(\cosh(2r )-1)\notag+\\
&\quad+\Re[\e^{-2i\theta}(\uphi)_{11}^2]\sinh(2r )\Big)\equiv\Delta_\varphi.
\end{align}
After applying the inverse Fourier transformation on the Gaussian characteristic function~\eqref{eq:CharacteristicFunction}, the probability distribution reads
\begin{equation}
p(x|\varphi) = \dfrac{1}{2\pi}\int \mathrm{d}\xi\, \chi(\xi|\varphi)\e^{\ii x\xi} =\dfrac{1}{\sqrt{2\pi\Delta_\varphi}}\exp\left(-\frac{x^2}{2\Delta_\varphi}\right)
\end{equation}
as displayed in~\eqref{eq:ProbabilityDistribution}.

\subsection{Derivation of~\eqref{eq:Displacement} } \label{app:Displacement}

Exploiting the unitarity of $\e^{\ii\theta\hat{a}_1^\dag\hat{a}_1^{\vphantom{\dag}}}$, we can write
\begin{align}
\label{eq:ExponentAndUnitary}
\e^{\ii \theta \hat{a}_1^\dag\hat{a}_1}\e^{-\ii  \xi \hat{x}_1}\e^{-\ii \theta \hat{a}_1^\dag\hat{a}_1}
&=\exp (-\ii\xi \e^{\ii \theta \hat{a}_1^\dag\hat{a}_1}\hat{x}_1\e^{-\ii \theta \hat{a}_1^\dag\hat{a}_1}).
\end{align}
By using the definition ${\hat{x}_1 = (\hat{a}_1 + \hat{a}_1^\dag)/\sqrt{2}}$, the first-mode quadrature operator along $\theta$ reads
\begin{equation}
\e^{\ii \theta \hat{a}_1^\dag\hat{a}_1}\hat{x}_1\e^{-\ii \theta \hat{a}_1^\dag\hat{a}_1} = 
\frac{1}{\sqrt{2}} (\hat{a}_1(\theta) + \hat{a}_1(\theta)^\dag),
\label{eq:x1theta}
\end{equation}
where $\hat{a}_1(\theta) = \e^{\ii \theta \hat{a}_1^\dag\hat{a}_1}\hat{a}_1\e^{-\ii \theta \hat{a}_1^\dag\hat{a}_1}$ is the first-mode annihilation operator at time $\theta$ in the Heisenberg picture. The Heisenberg equation is obtained by taking the derivative of $\hat{a}_1(\theta)$ with respect to $\theta$, and reads
\begin{equation}
\frac{\d \hat{a}_1(\theta) }{d\theta} = \ii [\hat{a}_1(\theta)^\dag\hat{a}_1(\theta), \hat{a}_1(\theta)] = - \ii \hat{a}_1(\theta),
\end{equation}
since $[\hat{a}_1(\theta),\hat{a}_1(\theta)^\dag] = [\hat{a}_1,\hat{a}_1^\dag] =1$ and $[\hat{a}_1(\theta),\hat{a}_1(\theta)] = [\hat{a}_1,\hat{a}_1] =0$.
Therefore,
\begin{equation}
\hat{a}_1(\theta) = \e^{-\ii \theta} \hat{a}_1, \qquad \hat{a}_1(\theta)^\dag = \e^{-\ii \theta} \hat{a}_1^\dag,
\label{eq:Heispict}
\end{equation}
the second equality being obtained by taking the adjoint of the first. By plugging~\eqref{eq:Heispict} into~\eqref{eq:x1theta} one gets
\begin{equation}
\xi\e^{\ii \theta \hat{a}_1^\dag\hat{a}_1}\hat{x}_1\e^{-\ii \theta \hat{a}_1^\dag\hat{a}_1} =  \xi\cos\theta\, \hat{x}_1 + \xi\,\sin\theta   \hat{p}_1 = \vec{\xi}_\theta\cdot\vec{z},
\end{equation}
where the vectors $\vec{\xi_\theta}$ and $\vec{z}$ are given in~\eqref{eq:DefinitionXiAndZ}.
Inserting this expression into~\eqref{eq:ExponentAndUnitary}, we finally obtain~\eqref{eq:Displacement}.
\section{Derivation of the Fisher Information in~\eqref{eq:Fisher0.1}}\label{Appendix:Fisher}
In this appendix we will evaluate the FI in~\eqref{eq:Fisher0.1} from the expression in~\eqref{eq:FisherGaussian}. Let us recall that the variance $\Delta_\varphi$ of the Gaussian probability density function~\eqref{eq:ProbabilityDistribution} reads
\begin{multline}
\label{eq:DeltaAppendix}
\Delta_\varphi = \dfrac{1}{2}\big(1+\abs{(\uphi)_{11}}^2(\cosh{2r }-1)+\\
+\Re[\e^{-2i\theta}(\uphi)_{11}^2]\sinh 2r \big).
\end{multline}
The derivative of $\Delta_\varphi$ is written as a sum of two contributions
\begin{multline}
\label{eq:DerivativeDelta}
\partial_\varphi\Delta_\varphi = \dfrac{1}{2}(\cosh 2r  -1)\partial_\varphi\abs{(\uphi)_{11}}^2+\\
+\dfrac{1}{2}\sinh 2r \Re[\e^{-2i\theta}\partial\varphi(\uphi)_{11}^2].
\end{multline}
The derivative in the first contribution is thus evaluated
\begin{align}
\partial_\varphi \abs{\uphiff}^2 
&=\uphiff^{*}\duphiff + \uphiff \duphiff^{*}\notag\\
&= 2\Re[ \uphiff^*\duphiff]
\end{align}
while the derivative in the second contribution reads
\begin{align}
\Re&[\e^{-2i\theta}\partial_\varphi \uphiff^2] = 2\Re[\e^{-2i\theta}\uphiff\duphiff]
\end{align}
Then, defining $\gamma_\varphi$ as the phase of $\uphiff$, and recalling that $\sinh^2 r  = N$, ~\eqref{eq:DerivativeDelta} reads
\begin{align}
&\partial_\varphi \Delta_\varphi =\notag\\
&= \Re\left[\uphiff^*\duphiff\left((\cosh 2r  -1)+\e^{2i(\gamma_\varphi-\theta)}\sinh 2r \right)\right]\notag\\
&=2\Re\left[\uphiff^*\duphiff\left(N+\e^{2i(\gamma_\varphi-\theta)}\sqrt{N^2+N}\right)\right],
\label{eq:DerivativeDeltaFinal}
\end{align}
while~\eqref{eq:DeltaAppendix} can be written as
\begin{multline}
\Delta_\varphi = \dfrac{1}{2}\Big(1+2\abs{\uphiff}^2\times \\
\times\big(N+\cos(2\gamma_\varphi-2\theta)\sqrt{N^2+N}\big)\Big).
\label{eq:DeltaFinal}
\end{multline}
Inserting the expressions~\eqref{eq:DerivativeDeltaFinal} and~\eqref{eq:DeltaFinal} into~\eqref{eq:FisherGaussian}, we get the FI 
\begin{align}
&F\left(\varphi\right)=\\
&=8\left(\dfrac{\Re[\uphiff^*\duphiff\left(N\! + \e^{2i(\gamma_\varphi - \theta)}\sqrt{N^2\!+\!N}\right)]}{1+2\abs{(u_\varphi)_{11}}^2\left(N\!+\cos(2\gamma_\varphi-2\theta)\sqrt{N^2\!+N}\right)}\right)^2\!\!\!.
\end{align}
Moreover we can write $\uphiff = \e^{\ii\gamma_\varphi}\sqrt{P_\varphi}$, so that 
\begin{equation}
\uphiff^*\duphiff = \frac{\partial_\varphi P_\varphi}{2} +\ii  P_\varphi\partial_\varphi\gamma_\varphi.
\end{equation}
Since $P_\varphi$, $\gamma_\varphi$ and their derivatives are real, once we define the quantities
\begin{equation}
f(N):=N\left(1+\cos[2(\gamma_\varphi-\theta)]\sqrt{1+1/N} \right),
\end{equation}
\begin{equation}
h(N):=N\sin[2(\gamma_\varphi-\theta)]\sqrt{1+1/N},
\end{equation}
we easily obtain 
\begin{equation}
F(\varphi) = 2\left(\dfrac{(\partial_\varphi P_\varphi)f(N) -2P_\varphi(\partial_\varphi\gamma_\varphi) h(N)}{1+2P_\varphi f(N)}\right)^2
\end{equation}
as displayed in~\eqref{eq:Fisher0.1}.
\section{Analytic distribution of the pre-factor \eqref{eq:prefactor} in the Fisher information \eqref{eq:FisherPrefactorf} for generators with only two distinct eigenvalues}\label{app:explicitDist}
We will derive here the explicit form of the probability density function for the pre-factor ${\f = (U^\dagger G_\varphi U)_{11}^2}$ in \eqref{eq:prefactor} for a fixed generator $G_\varphi$ as $U$ is sampled from $\mathrm{U}(M)$ with the Haar measure. First of all, note that this distribution depends only on the eigenvalues of $G_\varphi$, which we denote with $g=(g_1,\dots,g_M)$, dropping the $\varphi$ subscript for notation simplicity. This can be seen using the spectral decomposition of $G_\varphi=V_\varphi^\dagger D_\varphi V_\varphi$, where $D_\varphi=\diag(g)$,  yielding:
\begin{equation}
	U^\dagger G_\varphi U=(V_\varphi U)^\dagger D_\varphi (V_\varphi U)\sd U^\dagger D_\varphi U,
\end{equation} 
where in the last step we used the invariance property of the Haar measure and we used the notation $\sd$ to say that the two random variables have the same distribution. In light of this remark, we have that
\begin{align}\notag
	\f&\sd	f(U,D_\varphi)\\ \notag
	&=(U^\dagger D_\varphi U)_{11}^2\\
	&=\left(\sum_{j=1}^M \abs{u_j}^2 g_j\right)^2,
	\label{eq:RandomVariablegj}
\end{align}
having defined the random vector $u=Ue_1$ obtained by the application of the random matrix $U\in\mathrm{U}(M)$ to the fixed basis vector $e_1=(1,\dots,0)^T\in\C^M$, where $\C^M$ denotes the set of $M$-tuples of complex numbers. We see that $f(U,D_\varphi)$ can be interpreted as a weighted average of the eigenvalues of $G_\varphi$ with random weights; these weights are given by the square modulus of the components of a random vector drawn from the unit sphere in $\C^M$ with the Haar measure. The distribution of this random variable can be quite complicated for a generic choice of the $G_\varphi$'s eigenvalues $g=(g_1,\dots,g_M)$. We will consider here the situation in which there are at most two distinct eigenvalues $g_1\geqslant g_2\geqslant 0$, i.e:
\begin{equation}
g=(\underbrace{g_1,\dots,g_1}_k, \underbrace{g_2,\dots,g_2}_{M-k})
\end{equation}
so that 
\begin{align}\notag
	f(U,D_\varphi)&=\left(g_1 \sum_{j=1}^k \abs{u_j}^2 +g_2 \sum_{j=k+1}^M \abs{u_j}^2\right)^2\\
	&=\left[(g_1-g_2) \sum_{j=1}^k \abs{u_j}^2 +g_2\right]^2
	\label{eq:RandomVariablejM}
\end{align}
where we used the normalization constraint
\begin{equation}
	\sum_{j=1}^M \abs{u_j}^2=1.
\end{equation}
In order to get the distribution of~\eqref{eq:RandomVariablejM}, let us first consider the random quantity $\tau(U)$ defined by the sum inside the brackets, namely:
\begin{equation}
	\tau(U)=\sum_{j=1}^k \abs{u_j}^2
\end{equation}
We start from the distribution $q(t)$ of $\tau(U)$, defined in such a way that $q(t)\d t$ is the probability to have
\begin{equation}\label{eq:kdimsphericalCap}
	t \leqslant \tau(U)= \sum_{j=1}^k \abs{u_j}^2 \leqslant t+\d t
\end{equation}
or, defining $x_{2j-1}:=\Re u_j$ and $x_{2j}:=\Im u_j$, the probability to have:
\begin{equation}\label{eq:kdimsphericalCapBis}
t \leqslant \sum_{j=1}^k (x_{2j-1}^2+ x_{2j}^2) \leqslant t+\d t.
\end{equation}
This probability can be interpreted as the geometrical surface of a $2k$-dimensional hyperspherical cap of a $(2M-1)$-dimensional hypersphere sitting in $\R^{2M}$. Using this interpretation, one then finds that~\cite{Neumann1929,vonNeumann2010}:
\begin{equation}\label{eq:vonNeumannDist}
	q(t)=\frac{(M-1)!}{(k-1)!(M-k-1)!}t^{k-1}(1-t)^{M-k-1}\chi_{[0,1]}(t),
\end{equation}
where 
\begin{equation}
	\chi_{[0,1]}(t)=\begin{cases}
	1 & 0\leqslant t\leqslant 1,\\
	0 &\text{otherwise}.
	\end{cases}
\end{equation}
Starting from the distribution~\eqref{eq:vonNeumannDist} of $\tau(U)$, the probability density function $p(x)$ of 
\begin{equation}
	f(U,\diag(g))=\left[(g_1-g_2)\tau(U)+g_2\right]^2,
\end{equation}
can be found with a change of variables to be:
\begin{equation}
	p(x)=\frac{1}{2\sqrt{x}\Delta g}q\left(\frac{\sqrt{x}-g_2}{\Delta g}\right),
\end{equation}
where $\Delta g:=g_1-g_2$. We then have explicitly:
\begin{equation}\label{eq:ThPDF}
	p(x)=C\frac{(g_1-\sqrt{x})^{M-k-1}(\sqrt{x}-g_2)^{k-1}}{\sqrt{x}}\chi_{[g_2,g_1]}(\sqrt{x})
\end{equation}
where $C$ is a normalization constant given by:
\begin{equation}
	C=\frac{1}{2(\Delta g)^{M-1}}\frac{(M-1)!}{(k-1)!(M-k-1)!}.
\end{equation}
This distribution is valid whenever $G_\varphi$ has only two distinct positive eigenvalues $g_1\geqslant g_2\geqslant 0$. Numerical results are compared with the probability density function~\eqref{eq:ThPDF} in \figurename{ \ref{fig:FIdist}}.

\section{Derivation of the average value \eqref{eq:avef} of the pre-factor in \eqref{eq:prefactor} for random unitary matrices U}\label{app:averagesU}
We collect here some results needed for the computation of averages over the unitary group which have been used in the main text. Denoting with $\mathcal{P}$ the Haar probability measure,  the average of a function $f:\mathrm{U}(M)\rightarrow \C$ is defined as 
\begin{equation}
	\mathbb{E}[f(U)]=\int f(U)\ud \mathcal{P}(U)
\end{equation}
whenever this integral is defined. In order to derive the results of this work we are interested only in the moments of the matrix elements $U_{ij}$, i.e. the averages of some powers of the matrix elements and their complex conjugates, which we give here in the following lemma.

\begin{lemma}[\cite{hiai2000semicircle}, Proposition 4.2.3]
	Denoting with $U_{ij}$ the matrix elements of a Haar distributed random unitary matrix $U\in\mathrm{U}(M)$, the only non-vanishing moments up to the fourth order are given by:
	\begin{align}
		 \mathbb{E}[\abs{U_{ij}}^2]&=\frac{1}{M}  &(1\leqslant i,j \leqslant M),\\
		 \mathbb{E}[\abs{U_{ij}}^4]&=\frac{2}{M(M+1)} &(1\leqslant i,j \leqslant M),\\ 	
		 \mathbb{E}[\abs{U_{ij}}^2 \abs{U_{kj}}^2]&=\frac{1}{M(M+1)} & (i\neq k),\\
		 \mathbb{E}[\abs{U_{ij}}^2 \abs{U_{il}}^2]&=\frac{1}{M(M+1)} & (j\neq l),\\	
		 \mathbb{E}[\abs{U_{ij}}^2\abs{U_{kl}}^2]&=\frac{1}{M^2-1} & (i\neq k,\ j\neq l),\\
		 \mathbb{E}[U_{ij} U_{kl} U_{il}^* U_{kj}^*]&=-\frac{1}{M(M^2-1)} & (i\neq k,\ j\neq l).
	\end{align}
\end{lemma}

The results gathered in the previous lemma can be expressed with two compact formulas~\cite{puchala2017symbolic} as:
\begin{align}\label{eq:UHaarAverage1}
	&\mathbb{E}[U_{ij}U_{kl}^*]=\frac{\delta_{ik}\delta_{jl}}{M}\\
	&\notag \mathbb{E}[U_{ij}U_{kl}U_{mn}^* U_{pq}^*]=\frac{\delta_{im}\delta_{jn}\delta_{kp}\delta_{lq}+\delta_{ip}\delta_{jq}\delta_{km}\delta_{ln}}{M^2-1}+\\&\qquad \qquad-\frac{\delta_{im}\delta_{jq}\delta_{kp}\delta_{ln}+\delta_{ip}\delta_{jn}\delta_{km}\delta_{lq}}{M(M^2-1)}
	\label{eq:UHaarAverage2}
\end{align}
These formulas allow us to prove the following lemma, which has been used to derive the averages~\eqref{eq:avef} and~\eqref{eq:lowbound} of the main text.
\begin{lemma}
	Given a generic $M\times M$ complex matrix $A$, the following results hold:
	\begin{equation}\label{eq:aveUAU}
	\mathbb{E}[(U^\dagger A U)_{ij}]=\frac{\Tr(A)}{M}\delta_{ij}
	\end{equation}
	\begin{equation}\label{eq:aveUAUij2}
	\mathbb{E}[(U^\dagger A U)_{ij}^2]=\frac{\Tr(A)^2+\Tr(A^2)}{M(M+1)}\delta_{ij}
	\end{equation}
\end{lemma}
\begin{proof}
	Equation~\eqref{eq:aveUAU} can be proved using~\eqref{eq:UHaarAverage1}:
	\begin{align}\notag
	\mathbb{E}[(U^\dagger A U)_{ij}]&=\sum_{k,l}\mathbb{E}[U^\dagger_{ik}A_{kl}U_{lj}]\\
	\notag
	&=\sum_{k,l}A_{kl} \mathbb{E}[U^*_{ki}U_{lj}]\\
	\notag
	&=\sum_{k,l}A_{kl} \frac{\delta_{kl}\delta_{ij}}{M}\\
	&=\frac{\Tr(A)}{M}\delta_{ij},
	\end{align}
	while equation~\eqref{eq:aveUAUij2} can be proved using~\eqref{eq:UHaarAverage2}:
	\begin{align}\notag
	&\mathbb{E}[(U^\dagger A U)_{ij}^2]=\\
	\notag
	&\qquad=\sum_{k,l,m,n} \mathbb{E}[U^\dagger_{ik}A_{kl}U_{lj} U^\dagger_{im}A_{mn}U_{nj}]\\
	\notag
	&\qquad=\sum_{k,l,m,n}A_{kl}A_{mn} \mathbb{E}[U_{lj} U_{nj}U_{mi}^*U_{ki}^*]\\\notag
	&\qquad=[\Tr(A^2)+\Tr(A)^2]\left(\frac{1}{M^2-1}-\frac{1}{M(M^2-1)}\right)\delta_{ij}\\
	&\qquad=\frac{\Tr(A^2)+\Tr(A)^2}{M(M+1)}\delta_{ij}
	\end{align}
\end{proof}
For $A=G_\varphi$ and $i=j=1$, the expressions \eqref{eq:aveUAU} and \eqref{eq:aveUAUij2} reduces to the equalities in \eqref{eq:lowbound} and to \eqref{eq:avef}, respectively.
\section{Derivation of the typicality result in ~\eqref{eq:concentration}}\label{app:concentration}
In this appendix we will show how to derive equation~\eqref{eq:concentration} starting from a standard result on concentration of measure in high-dimensional spaces known as Levy's Lemma, which we report in the following theorem for the sake of completeness.

\begin{thm}\label{thm:Levy}
	Let $f:\mathbb{S}^{n-1}\rightarrow \R$ be a function defined over the unit euclidean sphere
	\begin{equation}
	\mathbb{S}^{n-1}=\bigg\{ x\in\R^n\bigg|\sum_{k=1}^n x_k^2=1\bigg\}
	\end{equation}
	endowed with the invariant Haar probability measure $\mathcal{P}$. Denote with  $L$ the Lipschitz constant of the function, i.e. the minimum $L$ such that
	\begin{equation}
	\abs{f(x)-f(y)}\leqslant L \Norm{x-y}_2,
	\end{equation}
	for all $x,y\in\mathbb{S}^{n-1}$, where $\Norm{x}_2=\sqrt{\sum_{k=1}^n x_k^2}$ is the Euclidean norm.
	Then:
	\begin{equation}\label{eq:LevyCor}
	\mathcal{P}(\abs{f-\E[f]}\geqslant \varepsilon )\leqslant 2 \mathrm{e}^{-\frac{n \varepsilon^2}{CL^2}},
	\end{equation}
	where $C$ is some positive constant which can be taken to be $C=9\pi^3$~\cite{facchi2017quantum,popescu2006entanglement}.
\end{thm}

In order to apply Theorem \ref{thm:Levy} to our case, we need to compute the Lipschitz constant associated with the pre-factor~\eqref{eq:prefactor}. First, note that $\f$ can be interpreted as a function defined on a real unit sphere. In fact, it can be written as
\begin{align}
\f
&=\left(\sum_{j=1}^M \abs{u_j}^2 g_j \right)^2.
\label{eq:F(u)}
\end{align}
where $u$ is a \textit{complex} vector on the unit sphere, given by $u=Ue$ with $e=(1,0,\dots,0)^T\in\C^M$. Since only the squared moduli $\abs{u_j}^2$ appear in this expression, we can recast the problem in terms of a real vector $x\in \R^{2M}$ whose components are defined by:
\begin{equation}\label{map_u_to_x}
x_{2j-1}=\Re u_j, \qquad x_{2j}=\Im u_j, \qquad j=1,\dots,M.
\end{equation}
The normalization constraint $\sum_{j=1}^M \abs{u_j}^2=1$ becomes 
\begin{equation}
\sum_{j=1}^{2M} x_j^2=1,
\end{equation} 
so that $x\in\mathbb{S}^{2M-1}$, the unit sphere sitting inside $\R^{2M}$. We see then that the random factor in equation~\eqref{eq:F(u)} can be envisioned as a function defined over the unit sphere $\mathbb{S}^{2M-1}$:
\begin{align}\notag
\left(\sum_{j=1}^{M} \abs{u_j}^2g_j\right)^2&=\left(\sum_{j=1}^{M} (x_{2j-1}^2+x_{2j}^2)g_j\right)^2\\
&=\left(x^T \widetilde{G} x\right)^2=:f(x),
\end{align}
where we have defined the diagonal matrix $\widetilde{G}=\diag(\widetilde{g})$ with $\widetilde{g}=(g_1,g_1,\dots,g_M,g_M)\in\R^{2M}$.
In order to apply Theorem \ref{thm:Levy} we need to estimate the Lipschitz constant $L$ of the function $f$; to this aim, we  evaluate the gradient of $f$, which is given by:
\begin{align}
\nabla f(x)
&=4(x^T\widetilde{G}x) \widetilde{G}x.
\end{align}
The Lipschitz constant for $f$ can be then obtained as:
\begin{equation}\label{eq:l}
L=\max_{x\in\mathbb{S}^{2M-1}} \Norm{\nabla f(x)}_2=4\Norm{G_\varphi}^2.
\end{equation}
To see this, note simply that:
\begin{align}\notag
\Norm{\nabla f(x)}_2&=\sqrt{[\nabla f(x)]^T [\nabla f(x)]}\\\notag
&=4|x^T \widetilde{G} x|\sqrt{x^T \widetilde{G}^2 x}\\\notag
&\leqslant 4 \|\widetilde{G}\|^2\\
&= 4 \Norm{G_\varphi}^2
\label{eq:max}
\end{align}
where in the inequality we used the fact that $|x^T\widetilde{G} x| \leqslant \|\widetilde{G}\|$ and $x^T \widetilde{G}^2 x=\|\widetilde{G}x\|^2\leqslant \|\widetilde{G}\|^2$, while in the last equality we used the fact that $\|\widetilde{G}\|=\Norm{G_\varphi}$. The value $\Norm{\nabla f(x)}_2=4\Norm{G_\varphi}^2$ can be obtained with $x=(1,0\dots,0)^T$, which together with~\eqref{eq:max} proves~\eqref{eq:l}. Applying Theorem \ref{thm:Levy} to our case with $n=2M$ and $L=4\Norm{G_\varphi}^2$ finally yields equation~\eqref{eq:concentration}.

\nocite{*}
\bibliography{references}

\begin{thebibliography}{47}%
\makeatletter
\providecommand \@ifxundefined [1]{%
 \@ifx{#1\undefined}
}%
\providecommand \@ifnum [1]{%
 \ifnum #1\expandafter \@firstoftwo
 \else \expandafter \@secondoftwo
 \fi
}%
\providecommand \@ifx [1]{%
 \ifx #1\expandafter \@firstoftwo
 \else \expandafter \@secondoftwo
 \fi
}%
\providecommand \natexlab [1]{#1}%
\providecommand \enquote  [1]{``#1''}%
\providecommand \bibnamefont  [1]{#1}%
\providecommand \bibfnamefont [1]{#1}%
\providecommand \citenamefont [1]{#1}%
\providecommand \href@noop [0]{\@secondoftwo}%
\providecommand \href [0]{\begingroup \@sanitize@url \@href}%
\providecommand \@href[1]{\@@startlink{#1}\@@href}%
\providecommand \@@href[1]{\endgroup#1\@@endlink}%
\providecommand \@sanitize@url [0]{\catcode `\\12\catcode `\$12\catcode
  `\&12\catcode `\#12\catcode `\^12\catcode `\_12\catcode `\%12\relax}%
\providecommand \@@startlink[1]{}%
\providecommand \@@endlink[0]{}%
\providecommand \url  [0]{\begingroup\@sanitize@url \@url }%
\providecommand \@url [1]{\endgroup\@href {#1}{\urlprefix }}%
\providecommand \urlprefix  [0]{URL }%
\providecommand \Eprint [0]{\href }%
\providecommand \doibase [0]{http://dx.doi.org/}%
\providecommand \selectlanguage [0]{\@gobble}%
\providecommand \bibinfo  [0]{\@secondoftwo}%
\providecommand \bibfield  [0]{\@secondoftwo}%
\providecommand \translation [1]{[#1]}%
\providecommand \BibitemOpen [0]{}%
\providecommand \bibitemStop [0]{}%
\providecommand \bibitemNoStop [0]{.\EOS\space}%
\providecommand \EOS [0]{\spacefactor3000\relax}%
\providecommand \BibitemShut  [1]{\csname bibitem#1\endcsname}%
\let\auto@bib@innerbib\@empty
\bibitem [{\citenamefont {Giovannetti}\ \emph {et~al.}(2004)\citenamefont
  {Giovannetti}, \citenamefont {Lloyd},\ and\ \citenamefont
  {Maccone}}]{Giovannetti2004}%
  \BibitemOpen
  \bibfield  {author} {\bibinfo {author} {\bibfnamefont {Vittorio}\
  \bibnamefont {Giovannetti}}, \bibinfo {author} {\bibfnamefont {Seth}\
  \bibnamefont {Lloyd}}, \ and\ \bibinfo {author} {\bibfnamefont {Lorenzo}\
  \bibnamefont {Maccone}},\ }\bibfield  {title} {\enquote {\bibinfo {title}
  {Quantum-enhanced measurements: Beating the standard quantum limit},}\ }\href
  {\doibase 10.1126/science.1104149} {\bibfield  {journal} {\bibinfo  {journal}
  {Science}\ }\textbf {\bibinfo {volume} {306}},\ \bibinfo {pages} {1330--1336}
  (\bibinfo {year} {2004})}\BibitemShut {NoStop}%
\bibitem [{\citenamefont {Giovannetti}\ \emph {et~al.}(2006)\citenamefont
  {Giovannetti}, \citenamefont {Lloyd},\ and\ \citenamefont
  {Maccone}}]{Giovannetti2006}%
  \BibitemOpen
  \bibfield  {author} {\bibinfo {author} {\bibfnamefont {Vittorio}\
  \bibnamefont {Giovannetti}}, \bibinfo {author} {\bibfnamefont {Seth}\
  \bibnamefont {Lloyd}}, \ and\ \bibinfo {author} {\bibfnamefont {Lorenzo}\
  \bibnamefont {Maccone}},\ }\bibfield  {title} {\enquote {\bibinfo {title}
  {Quantum metrology},}\ }\href {\doibase 10.1103/PhysRevLett.96.010401}
  {\bibfield  {journal} {\bibinfo  {journal} {Phys. Rev. Lett.}\ }\textbf
  {\bibinfo {volume} {96}},\ \bibinfo {pages} {010401} (\bibinfo {year}
  {2006})}\BibitemShut {NoStop}%
\bibitem [{\citenamefont {Dowling}(2008)}]{Dowling2008}%
  \BibitemOpen
  \bibfield  {author} {\bibinfo {author} {\bibfnamefont {Jonathan~P.}\
  \bibnamefont {Dowling}},\ }\bibfield  {title} {\enquote {\bibinfo {title}
  {Quantum optical metrology – the lowdown on high-n00n states},}\ }\href
  {\doibase 10.1080/00107510802091298} {\bibfield  {journal} {\bibinfo
  {journal} {Contemporary Physics}\ }\textbf {\bibinfo {volume} {49}},\
  \bibinfo {pages} {125--143} (\bibinfo {year} {2008})}\BibitemShut {NoStop}%
\bibitem [{\citenamefont {Giovannetti}\ \emph {et~al.}(2011)\citenamefont
  {Giovannetti}, \citenamefont {Lloyd},\ and\ \citenamefont
  {Maccone}}]{Giovannetti2011}%
  \BibitemOpen
  \bibfield  {author} {\bibinfo {author} {\bibfnamefont {Vittorio}\
  \bibnamefont {Giovannetti}}, \bibinfo {author} {\bibfnamefont {Seth}\
  \bibnamefont {Lloyd}}, \ and\ \bibinfo {author} {\bibfnamefont {Lorenzo}\
  \bibnamefont {Maccone}},\ }\bibfield  {title} {\enquote {\bibinfo {title}
  {Advances in quantum metrology},}\ }\href {\doibase 10.1038/nphoton.2011.35}
  {\bibfield  {journal} {\bibinfo  {journal} {Nature Photonics}\ }\textbf
  {\bibinfo {volume} {5}},\ \bibinfo {pages} {010401} (\bibinfo {year}
  {2011})}\BibitemShut {NoStop}%
\bibitem [{\citenamefont {{Dowling}}\ and\ \citenamefont
  {{Seshadreesan}}(2015)}]{Dowling2015}%
  \BibitemOpen
  \bibfield  {author} {\bibinfo {author} {\bibfnamefont {J.~P.}\ \bibnamefont
  {{Dowling}}}\ and\ \bibinfo {author} {\bibfnamefont {K.~P.}\ \bibnamefont
  {{Seshadreesan}}},\ }\bibfield  {title} {\enquote {\bibinfo {title} {Quantum
  optical technologies for metrology, sensing, and imaging},}\ }\href {\doibase
  10.1109/JLT.2014.2386795} {\bibfield  {journal} {\bibinfo  {journal} {Journal
  of Lightwave Technology}\ }\textbf {\bibinfo {volume} {33}},\ \bibinfo
  {pages} {2359--2370} (\bibinfo {year} {2015})}\BibitemShut {NoStop}%
\bibitem [{\citenamefont {Bondurant}\ and\ \citenamefont
  {Shapiro}(1984)}]{Shapiro84}%
  \BibitemOpen
  \bibfield  {author} {\bibinfo {author} {\bibfnamefont {Roy~S.}\ \bibnamefont
  {Bondurant}}\ and\ \bibinfo {author} {\bibfnamefont {Jeffrey~H.}\
  \bibnamefont {Shapiro}},\ }\bibfield  {title} {\enquote {\bibinfo {title}
  {Squeezed states in phase-sensing interferometers},}\ }\href {\doibase
  10.1103/PhysRevD.30.2548} {\bibfield  {journal} {\bibinfo  {journal} {Phys.
  Rev. D}\ }\textbf {\bibinfo {volume} {30}},\ \bibinfo {pages} {2548--2556}
  (\bibinfo {year} {1984})}\BibitemShut {NoStop}%
\bibitem [{\citenamefont {Wineland}\ \emph {et~al.}(1992)\citenamefont
  {Wineland}, \citenamefont {Bollinger}, \citenamefont {Itano}, \citenamefont
  {Moore},\ and\ \citenamefont {Heinzen}}]{Wineland92}%
  \BibitemOpen
  \bibfield  {author} {\bibinfo {author} {\bibfnamefont {D.~J.}\ \bibnamefont
  {Wineland}}, \bibinfo {author} {\bibfnamefont {J.~J.}\ \bibnamefont
  {Bollinger}}, \bibinfo {author} {\bibfnamefont {W.~M.}\ \bibnamefont
  {Itano}}, \bibinfo {author} {\bibfnamefont {F.~L.}\ \bibnamefont {Moore}}, \
  and\ \bibinfo {author} {\bibfnamefont {D.~J.}\ \bibnamefont {Heinzen}},\
  }\bibfield  {title} {\enquote {\bibinfo {title} {Spin squeezing and reduced
  quantum noise in spectroscopy},}\ }\href {\doibase 10.1103/PhysRevA.46.R6797}
  {\bibfield  {journal} {\bibinfo  {journal} {Phys. Rev. A}\ }\textbf {\bibinfo
  {volume} {46}},\ \bibinfo {pages} {R6797--R6800} (\bibinfo {year}
  {1992})}\BibitemShut {NoStop}%
\bibitem [{\citenamefont {{Maccone}}\ and\ \citenamefont
  {{Riccardi}}(2019)}]{Maccone2019}%
  \BibitemOpen
  \bibfield  {author} {\bibinfo {author} {\bibfnamefont {Lorenzo}\ \bibnamefont
  {{Maccone}}}\ and\ \bibinfo {author} {\bibfnamefont {Alberto}\ \bibnamefont
  {{Riccardi}}},\ }\bibfield  {title} {\enquote {\bibinfo {title} {{Squeezing
  metrology}},}\ }\href@noop {} {\ ,\ \bibinfo {pages} {arXiv:1901.07482}
  (\bibinfo {year} {2019})},\ \Eprint {http://arxiv.org/abs/1901.07482}
  {arXiv:1901.07482 [quant-ph]} \BibitemShut {NoStop}%
\bibitem [{\citenamefont {Helstrom}(1969)}]{Helstrom1969}%
  \BibitemOpen
  \bibfield  {author} {\bibinfo {author} {\bibfnamefont {Carl~W.}\ \bibnamefont
  {Helstrom}},\ }\bibfield  {title} {\enquote {\bibinfo {title} {Quantum
  detection and estimation theory},}\ }\href {\doibase 10.1007/BF01007479}
  {\bibfield  {journal} {\bibinfo  {journal} {Journal of Statistical Physics}\
  }\textbf {\bibinfo {volume} {1}},\ \bibinfo {pages} {231--252} (\bibinfo
  {year} {1969})}\BibitemShut {NoStop}%
\bibitem [{\citenamefont {Braunstein}\ and\ \citenamefont
  {Caves}(1994)}]{PhysRevLett.72.3439}%
  \BibitemOpen
  \bibfield  {author} {\bibinfo {author} {\bibfnamefont {Samuel~L.}\
  \bibnamefont {Braunstein}}\ and\ \bibinfo {author} {\bibfnamefont
  {Carlton~M.}\ \bibnamefont {Caves}},\ }\bibfield  {title} {\enquote {\bibinfo
  {title} {Statistical distance and the geometry of quantum states},}\ }\href
  {\doibase 10.1103/PhysRevLett.72.3439} {\bibfield  {journal} {\bibinfo
  {journal} {Phys. Rev. Lett.}\ }\textbf {\bibinfo {volume} {72}},\ \bibinfo
  {pages} {3439--3443} (\bibinfo {year} {1994})}\BibitemShut {NoStop}%
\bibitem [{\citenamefont {Seshadreesan}\ \emph {et~al.}(2011)\citenamefont
  {Seshadreesan}, \citenamefont {Anisimov}, \citenamefont {Lee},\ and\
  \citenamefont {Dowling}}]{Seshadreesan_2011}%
  \BibitemOpen
  \bibfield  {author} {\bibinfo {author} {\bibfnamefont {Kaushik~P}\
  \bibnamefont {Seshadreesan}}, \bibinfo {author} {\bibfnamefont {Petr~M}\
  \bibnamefont {Anisimov}}, \bibinfo {author} {\bibfnamefont {Hwang}\
  \bibnamefont {Lee}}, \ and\ \bibinfo {author} {\bibfnamefont {Jonathan~P}\
  \bibnamefont {Dowling}},\ }\bibfield  {title} {\enquote {\bibinfo {title}
  {Parity detection achieves the heisenberg limit in interferometry with
  coherent mixed with squeezed vacuum light},}\ }\href {\doibase
  10.1088/1367-2630/13/8/083026} {\bibfield  {journal} {\bibinfo  {journal}
  {New Journal of Physics}\ }\textbf {\bibinfo {volume} {13}},\ \bibinfo
  {pages} {083026} (\bibinfo {year} {2011})}\BibitemShut {NoStop}%
\bibitem [{\citenamefont {De~Pasquale}\ \emph {et~al.}(2015)\citenamefont
  {De~Pasquale}, \citenamefont {Facchi}, \citenamefont {Florio}, \citenamefont
  {Giovannetti}, \citenamefont {Matsuoka},\ and\ \citenamefont
  {Yuasa}}]{PhysRevA.92.042115}%
  \BibitemOpen
  \bibfield  {author} {\bibinfo {author} {\bibfnamefont {Antonella}\
  \bibnamefont {De~Pasquale}}, \bibinfo {author} {\bibfnamefont {Paolo}\
  \bibnamefont {Facchi}}, \bibinfo {author} {\bibfnamefont {Giuseppe}\
  \bibnamefont {Florio}}, \bibinfo {author} {\bibfnamefont {Vittorio}\
  \bibnamefont {Giovannetti}}, \bibinfo {author} {\bibfnamefont {Koji}\
  \bibnamefont {Matsuoka}}, \ and\ \bibinfo {author} {\bibfnamefont {Kazuya}\
  \bibnamefont {Yuasa}},\ }\bibfield  {title} {\enquote {\bibinfo {title}
  {Two-mode bosonic quantum metrology with number fluctuations},}\ }\href
  {\doibase 10.1103/PhysRevA.92.042115} {\bibfield  {journal} {\bibinfo
  {journal} {Phys. Rev. A}\ }\textbf {\bibinfo {volume} {92}},\ \bibinfo
  {pages} {042115} (\bibinfo {year} {2015})}\BibitemShut {NoStop}%
\bibitem [{\citenamefont {Uys}\ and\ \citenamefont
  {Meystre}(2007)}]{PhysRevA.76.013804}%
  \BibitemOpen
  \bibfield  {author} {\bibinfo {author} {\bibfnamefont {H.}~\bibnamefont
  {Uys}}\ and\ \bibinfo {author} {\bibfnamefont {P.}~\bibnamefont {Meystre}},\
  }\bibfield  {title} {\enquote {\bibinfo {title} {Quantum states for
  heisenberg-limited interferometry},}\ }\href {\doibase
  10.1103/PhysRevA.76.013804} {\bibfield  {journal} {\bibinfo  {journal} {Phys.
  Rev. A}\ }\textbf {\bibinfo {volume} {76}},\ \bibinfo {pages} {013804}
  (\bibinfo {year} {2007})}\BibitemShut {NoStop}%
\bibitem [{\citenamefont {Lang}\ and\ \citenamefont
  {Caves}(2014)}]{PhysRevA.90.025802}%
  \BibitemOpen
  \bibfield  {author} {\bibinfo {author} {\bibfnamefont {Matthias~D.}\
  \bibnamefont {Lang}}\ and\ \bibinfo {author} {\bibfnamefont {Carlton~M.}\
  \bibnamefont {Caves}},\ }\bibfield  {title} {\enquote {\bibinfo {title}
  {Optimal quantum-enhanced interferometry},}\ }\href {\doibase
  10.1103/PhysRevA.90.025802} {\bibfield  {journal} {\bibinfo  {journal} {Phys.
  Rev. A}\ }\textbf {\bibinfo {volume} {90}},\ \bibinfo {pages} {025802}
  (\bibinfo {year} {2014})}\BibitemShut {NoStop}%
\bibitem [{\citenamefont {Weedbrook}\ \emph {et~al.}(2012)\citenamefont
  {Weedbrook}, \citenamefont {Pirandola}, \citenamefont {Garc\'{\i}a-Patr\'on},
  \citenamefont {Cerf}, \citenamefont {Ralph}, \citenamefont {Shapiro},\ and\
  \citenamefont {Lloyd}}]{Weedbrook2012}%
  \BibitemOpen
  \bibfield  {author} {\bibinfo {author} {\bibfnamefont {Christian}\
  \bibnamefont {Weedbrook}}, \bibinfo {author} {\bibfnamefont {Stefano}\
  \bibnamefont {Pirandola}}, \bibinfo {author} {\bibfnamefont {Ra\'ul}\
  \bibnamefont {Garc\'{\i}a-Patr\'on}}, \bibinfo {author} {\bibfnamefont
  {Nicolas~J.}\ \bibnamefont {Cerf}}, \bibinfo {author} {\bibfnamefont
  {Timothy~C.}\ \bibnamefont {Ralph}}, \bibinfo {author} {\bibfnamefont
  {Jeffrey~H.}\ \bibnamefont {Shapiro}}, \ and\ \bibinfo {author}
  {\bibfnamefont {Seth}\ \bibnamefont {Lloyd}},\ }\bibfield  {title} {\enquote
  {\bibinfo {title} {Gaussian quantum information},}\ }\href {\doibase
  10.1103/RevModPhys.84.621} {\bibfield  {journal} {\bibinfo  {journal} {Rev.
  Mod. Phys.}\ }\textbf {\bibinfo {volume} {84}},\ \bibinfo {pages} {621--669}
  (\bibinfo {year} {2012})}\BibitemShut {NoStop}%
\bibitem [{\citenamefont {Adesso}\ \emph {et~al.}(2014)\citenamefont {Adesso},
  \citenamefont {Ragy},\ and\ \citenamefont {Lee}}]{Adesso2014}%
  \BibitemOpen
  \bibfield  {author} {\bibinfo {author} {\bibfnamefont {Gerardo}\ \bibnamefont
  {Adesso}}, \bibinfo {author} {\bibfnamefont {Sammy}\ \bibnamefont {Ragy}}, \
  and\ \bibinfo {author} {\bibfnamefont {Antony~R.}\ \bibnamefont {Lee}},\
  }\bibfield  {title} {\enquote {\bibinfo {title} {Continuous variable quantum
  information: Gaussian states and beyond},}\ }\href {\doibase
  10.1142/S1230161214400010} {\bibfield  {journal} {\bibinfo  {journal} {Open
  Systems \& Information Dynamics}\ }\textbf {\bibinfo {volume} {21}},\
  \bibinfo {pages} {1440001} (\bibinfo {year} {2014})},\ \Eprint
  {http://arxiv.org/abs/https://doi.org/10.1142/S1230161214400010}
  {https://doi.org/10.1142/S1230161214400010} \BibitemShut {NoStop}%
\bibitem [{\citenamefont {Ferraro}\ \emph {et~al.}(2005)\citenamefont
  {Ferraro}, \citenamefont {Olivares},\ and\ \citenamefont
  {Paris}}]{Paris2005}%
  \BibitemOpen
  \bibfield  {author} {\bibinfo {author} {\bibfnamefont {A.}~\bibnamefont
  {Ferraro}}, \bibinfo {author} {\bibfnamefont {S.}~\bibnamefont {Olivares}}, \
  and\ \bibinfo {author} {\bibfnamefont {{Matteo G. A.}}\ \bibnamefont
  {Paris}},\ }\href@noop {} {\emph {\bibinfo {title} {Gaussian States in
  Quantum Information}}},\ Napoli Series on physics and Astrophysics\ (\bibinfo
   {publisher} {Bibliopolis},\ \bibinfo {year} {2005})\BibitemShut {NoStop}%
\bibitem [{\citenamefont {Monras}(2006)}]{monras2006}%
  \BibitemOpen
  \bibfield  {author} {\bibinfo {author} {\bibfnamefont {Alex}\ \bibnamefont
  {Monras}},\ }\bibfield  {title} {\enquote {\bibinfo {title} {Optimal phase
  measurements with pure gaussian states},}\ }\href {\doibase
  10.1103/PhysRevA.73.033821} {\bibfield  {journal} {\bibinfo  {journal} {Phys.
  Rev. A}\ }\textbf {\bibinfo {volume} {73}},\ \bibinfo {pages} {033821}
  (\bibinfo {year} {2006})}\BibitemShut {NoStop}%
\bibitem [{\citenamefont {Matsubara}\ \emph {et~al.}(2019)\citenamefont
  {Matsubara}, \citenamefont {Facchi}, \citenamefont {Giovannetti},\ and\
  \citenamefont {Yuasa}}]{Matsubara_2019}%
  \BibitemOpen
  \bibfield  {author} {\bibinfo {author} {\bibfnamefont {Teruo}\ \bibnamefont
  {Matsubara}}, \bibinfo {author} {\bibfnamefont {Paolo}\ \bibnamefont
  {Facchi}}, \bibinfo {author} {\bibfnamefont {Vittorio}\ \bibnamefont
  {Giovannetti}}, \ and\ \bibinfo {author} {\bibfnamefont {Kazuya}\
  \bibnamefont {Yuasa}},\ }\bibfield  {title} {\enquote {\bibinfo {title}
  {Optimal gaussian metrology for generic multimode interferometric circuit},}\
  }\href {\doibase 10.1088/1367-2630/ab0604} {\bibfield  {journal} {\bibinfo
  {journal} {New Journal of Physics}\ }\textbf {\bibinfo {volume} {21}},\
  \bibinfo {pages} {033014} (\bibinfo {year} {2019})}\BibitemShut {NoStop}%
\bibitem [{\citenamefont {Oh}\ \emph {et~al.}(2019)\citenamefont {Oh},
  \citenamefont {Lee}, \citenamefont {Rockstuhl}, \citenamefont {Jeong},
  \citenamefont {Kim}, \citenamefont {Nha},\ and\ \citenamefont
  {Lee}}]{Oh2019}%
  \BibitemOpen
  \bibfield  {author} {\bibinfo {author} {\bibfnamefont {Changhun}\
  \bibnamefont {Oh}}, \bibinfo {author} {\bibfnamefont {Changhyoup}\
  \bibnamefont {Lee}}, \bibinfo {author} {\bibfnamefont {Carsten}\ \bibnamefont
  {Rockstuhl}}, \bibinfo {author} {\bibfnamefont {Hyunseok}\ \bibnamefont
  {Jeong}}, \bibinfo {author} {\bibfnamefont {Jaewan}\ \bibnamefont {Kim}},
  \bibinfo {author} {\bibfnamefont {Hyunchul}\ \bibnamefont {Nha}}, \ and\
  \bibinfo {author} {\bibfnamefont {Su-Yong}\ \bibnamefont {Lee}},\ }\bibfield
  {title} {\enquote {\bibinfo {title} {Optimal gaussian measurements for phase
  estimation in single-mode gaussian metrology},}\ }\href {\doibase
  10.1038/s41534-019-0124-4} {\bibfield  {journal} {\bibinfo  {journal} {npj
  Quantum Information}\ }\textbf {\bibinfo {volume} {5}},\ \bibinfo {pages}
  {10} (\bibinfo {year} {2019})}\BibitemShut {NoStop}%
\bibitem [{\citenamefont {Wiseman}(1995)}]{adaptiveHomodyne1995}%
  \BibitemOpen
  \bibfield  {author} {\bibinfo {author} {\bibfnamefont {H.~M.}\ \bibnamefont
  {Wiseman}},\ }\bibfield  {title} {\enquote {\bibinfo {title} {Adaptive phase
  measurements of optical modes: Going beyond the marginal $q$ distribution},}\
  }\href {\doibase 10.1103/PhysRevLett.75.4587} {\bibfield  {journal} {\bibinfo
   {journal} {Phys. Rev. Lett.}\ }\textbf {\bibinfo {volume} {75}},\ \bibinfo
  {pages} {4587--4590} (\bibinfo {year} {1995})}\BibitemShut {NoStop}%
\bibitem [{\citenamefont {Armen}\ \emph {et~al.}(2002)\citenamefont {Armen},
  \citenamefont {Au}, \citenamefont {Stockton}, \citenamefont {Doherty},\ and\
  \citenamefont {Mabuchi}}]{adaptiveHomodyne2002}%
  \BibitemOpen
  \bibfield  {author} {\bibinfo {author} {\bibfnamefont {Michael~A.}\
  \bibnamefont {Armen}}, \bibinfo {author} {\bibfnamefont {John~K.}\
  \bibnamefont {Au}}, \bibinfo {author} {\bibfnamefont {John~K.}\ \bibnamefont
  {Stockton}}, \bibinfo {author} {\bibfnamefont {Andrew~C.}\ \bibnamefont
  {Doherty}}, \ and\ \bibinfo {author} {\bibfnamefont {Hideo}\ \bibnamefont
  {Mabuchi}},\ }\bibfield  {title} {\enquote {\bibinfo {title} {Adaptive
  homodyne measurement of optical phase},}\ }\href {\doibase
  10.1103/PhysRevLett.89.133602} {\bibfield  {journal} {\bibinfo  {journal}
  {Phys. Rev. Lett.}\ }\textbf {\bibinfo {volume} {89}},\ \bibinfo {pages}
  {133602} (\bibinfo {year} {2002})}\BibitemShut {NoStop}%
\bibitem [{\citenamefont {Aspachs}\ \emph {et~al.}(2009)\citenamefont
  {Aspachs}, \citenamefont {Calsamiglia}, \citenamefont {Mu\~noz Tapia},\ and\
  \citenamefont {Bagan}}]{Aspachs2009}%
  \BibitemOpen
  \bibfield  {author} {\bibinfo {author} {\bibfnamefont {M.}~\bibnamefont
  {Aspachs}}, \bibinfo {author} {\bibfnamefont {J.}~\bibnamefont
  {Calsamiglia}}, \bibinfo {author} {\bibfnamefont {R.}~\bibnamefont {Mu\~noz
  Tapia}}, \ and\ \bibinfo {author} {\bibfnamefont {E.}~\bibnamefont {Bagan}},\
  }\bibfield  {title} {\enquote {\bibinfo {title} {Phase estimation for thermal
  gaussian states},}\ }\href {\doibase 10.1103/PhysRevA.79.033834} {\bibfield
  {journal} {\bibinfo  {journal} {Phys. Rev. A}\ }\textbf {\bibinfo {volume}
  {79}},\ \bibinfo {pages} {033834} (\bibinfo {year} {2009})}\BibitemShut
  {NoStop}%
\bibitem [{\citenamefont {Gatto}\ \emph {et~al.}(2019)\citenamefont {Gatto},
  \citenamefont {Facchi}, \citenamefont {Narducci},\ and\ \citenamefont
  {Tamma}}]{Gatto2019}%
  \BibitemOpen
  \bibfield  {author} {\bibinfo {author} {\bibfnamefont {Dario}\ \bibnamefont
  {Gatto}}, \bibinfo {author} {\bibfnamefont {Paolo}\ \bibnamefont {Facchi}},
  \bibinfo {author} {\bibfnamefont {Frank~A.}\ \bibnamefont {Narducci}}, \ and\
  \bibinfo {author} {\bibfnamefont {Vincenzo}\ \bibnamefont {Tamma}},\
  }\bibfield  {title} {\enquote {\bibinfo {title} {Distributed quantum
  metrology with a single squeezed-vacuum source},}\ }\href {\doibase
  10.1103/PhysRevResearch.1.032024} {\bibfield  {journal} {\bibinfo  {journal}
  {Phys. Rev. Research}\ }\textbf {\bibinfo {volume} {1}},\ \bibinfo {pages}
  {032024} (\bibinfo {year} {2019})}\BibitemShut {NoStop}%
\bibitem [{\citenamefont {Gatto}\ \emph {et~al.}(0)\citenamefont {Gatto},
  \citenamefont {Facchi},\ and\ \citenamefont {Tamma}}]{Gatto2020}%
  \BibitemOpen
  \bibfield  {author} {\bibinfo {author} {\bibfnamefont {Dario}\ \bibnamefont
  {Gatto}}, \bibinfo {author} {\bibfnamefont {Paolo}\ \bibnamefont {Facchi}}, \
  and\ \bibinfo {author} {\bibfnamefont {Vincenzo}\ \bibnamefont {Tamma}},\
  }\bibfield  {title} {\enquote {\bibinfo {title} {Phase space
  heisenberg-limited estimation of the average phase shift in a mach–zehnder
  interferometer},}\ }\href {\doibase 10.1142/S0219749919410193} {\bibfield
  {journal} {\bibinfo  {journal} {International Journal of Quantum
  Information}\ }\textbf {\bibinfo {volume} {0}},\ \bibinfo {pages} {1941019}
  (\bibinfo {year} {0})},\ \Eprint
  {http://arxiv.org/abs/https://doi.org/10.1142/S0219749919410193}
  {https://doi.org/10.1142/S0219749919410193} \BibitemShut {NoStop}%
\bibitem [{\citenamefont {Takeoka}\ \emph {et~al.}(2017)\citenamefont
  {Takeoka}, \citenamefont {Seshadreesan}, \citenamefont {You}, \citenamefont
  {Izumi},\ and\ \citenamefont {Dowling}}]{PhysRevA.96.052118}%
  \BibitemOpen
  \bibfield  {author} {\bibinfo {author} {\bibfnamefont {Masahiro}\
  \bibnamefont {Takeoka}}, \bibinfo {author} {\bibfnamefont {Kaushik~P.}\
  \bibnamefont {Seshadreesan}}, \bibinfo {author} {\bibfnamefont {Chenglong}\
  \bibnamefont {You}}, \bibinfo {author} {\bibfnamefont {Shuro}\ \bibnamefont
  {Izumi}}, \ and\ \bibinfo {author} {\bibfnamefont {Jonathan~P.}\ \bibnamefont
  {Dowling}},\ }\bibfield  {title} {\enquote {\bibinfo {title} {Fundamental
  precision limit of a mach-zehnder interferometric sensor when one of the
  inputs is the vacuum},}\ }\href {\doibase 10.1103/PhysRevA.96.052118}
  {\bibfield  {journal} {\bibinfo  {journal} {Phys. Rev. A}\ }\textbf {\bibinfo
  {volume} {96}},\ \bibinfo {pages} {052118} (\bibinfo {year}
  {2017})}\BibitemShut {NoStop}%
\bibitem [{\citenamefont {Ge}\ \emph {et~al.}(2018)\citenamefont {Ge},
  \citenamefont {Jacobs}, \citenamefont {Eldredge}, \citenamefont {Gorshkov},\
  and\ \citenamefont {Foss-Feig}}]{Ge2018}%
  \BibitemOpen
  \bibfield  {author} {\bibinfo {author} {\bibfnamefont {Wenchao}\ \bibnamefont
  {Ge}}, \bibinfo {author} {\bibfnamefont {Kurt}\ \bibnamefont {Jacobs}},
  \bibinfo {author} {\bibfnamefont {Zachary}\ \bibnamefont {Eldredge}},
  \bibinfo {author} {\bibfnamefont {Alexey~V.}\ \bibnamefont {Gorshkov}}, \
  and\ \bibinfo {author} {\bibfnamefont {Michael}\ \bibnamefont {Foss-Feig}},\
  }\bibfield  {title} {\enquote {\bibinfo {title} {Distributed quantum
  metrology with linear networks and separable inputs},}\ }\href {\doibase
  10.1103/PhysRevLett.121.043604} {\bibfield  {journal} {\bibinfo  {journal}
  {Phys. Rev. Lett.}\ }\textbf {\bibinfo {volume} {121}},\ \bibinfo {pages}
  {043604} (\bibinfo {year} {2018})}\BibitemShut {NoStop}%
\bibitem [{\citenamefont {Zhuang}\ \emph {et~al.}(2018)\citenamefont {Zhuang},
  \citenamefont {Zhang},\ and\ \citenamefont {Shapiro}}]{Zhuang2018}%
  \BibitemOpen
  \bibfield  {author} {\bibinfo {author} {\bibfnamefont {Quntao}\ \bibnamefont
  {Zhuang}}, \bibinfo {author} {\bibfnamefont {Zheshen}\ \bibnamefont {Zhang}},
  \ and\ \bibinfo {author} {\bibfnamefont {Jeffrey~H.}\ \bibnamefont
  {Shapiro}},\ }\bibfield  {title} {\enquote {\bibinfo {title} {Distributed
  quantum sensing using continuous-variable multipartite entanglement},}\
  }\href {\doibase 10.1103/PhysRevA.97.032329} {\bibfield  {journal} {\bibinfo
  {journal} {Phys. Rev. A}\ }\textbf {\bibinfo {volume} {97}},\ \bibinfo
  {pages} {032329} (\bibinfo {year} {2018})}\BibitemShut {NoStop}%
\bibitem [{\citenamefont {Sidhu}\ and\ \citenamefont {Kok}(2020)}]{Sidhu2020}%
  \BibitemOpen
  \bibfield  {author} {\bibinfo {author} {\bibfnamefont {Jasminder~S}\
  \bibnamefont {Sidhu}}\ and\ \bibinfo {author} {\bibfnamefont {Pieter}\
  \bibnamefont {Kok}},\ }\bibfield  {title} {\enquote {\bibinfo {title}
  {Geometric perspective on quantum parameter estimation},}\ }\href@noop {}
  {\bibfield  {journal} {\bibinfo  {journal} {AVS Quantum Science}\ }\textbf
  {\bibinfo {volume} {2}},\ \bibinfo {pages} {014701} (\bibinfo {year}
  {2020})}\BibitemShut {NoStop}%
\bibitem [{\citenamefont {Qian}\ \emph {et~al.}(2019)\citenamefont {Qian},
  \citenamefont {Eldredge}, \citenamefont {Ge}, \citenamefont {Pagano},
  \citenamefont {Monroe}, \citenamefont {Porto},\ and\ \citenamefont
  {Gorshkov}}]{Qian2019}%
  \BibitemOpen
  \bibfield  {author} {\bibinfo {author} {\bibfnamefont {Kevin}\ \bibnamefont
  {Qian}}, \bibinfo {author} {\bibfnamefont {Zachary}\ \bibnamefont
  {Eldredge}}, \bibinfo {author} {\bibfnamefont {Wenchao}\ \bibnamefont {Ge}},
  \bibinfo {author} {\bibfnamefont {Guido}\ \bibnamefont {Pagano}}, \bibinfo
  {author} {\bibfnamefont {Christopher}\ \bibnamefont {Monroe}}, \bibinfo
  {author} {\bibfnamefont {J.~V.}\ \bibnamefont {Porto}}, \ and\ \bibinfo
  {author} {\bibfnamefont {Alexey~V.}\ \bibnamefont {Gorshkov}},\ }\bibfield
  {title} {\enquote {\bibinfo {title} {Heisenberg-scaling measurement protocol
  for analytic functions with quantum sensor networks},}\ }\href {\doibase
  10.1103/PhysRevA.100.042304} {\bibfield  {journal} {\bibinfo  {journal}
  {Phys. Rev. A}\ }\textbf {\bibinfo {volume} {100}},\ \bibinfo {pages}
  {042304} (\bibinfo {year} {2019})}\BibitemShut {NoStop}%
\bibitem [{\citenamefont {Xia}\ \emph {et~al.}(2020)\citenamefont {Xia},
  \citenamefont {Li}, \citenamefont {Clark}, \citenamefont {Hart},
  \citenamefont {Zhuang},\ and\ \citenamefont {Zhang}}]{Xia2020}%
  \BibitemOpen
  \bibfield  {author} {\bibinfo {author} {\bibfnamefont {Yi}~\bibnamefont
  {Xia}}, \bibinfo {author} {\bibfnamefont {Wei}\ \bibnamefont {Li}}, \bibinfo
  {author} {\bibfnamefont {William}\ \bibnamefont {Clark}}, \bibinfo {author}
  {\bibfnamefont {Darlene}\ \bibnamefont {Hart}}, \bibinfo {author}
  {\bibfnamefont {Quntao}\ \bibnamefont {Zhuang}}, \ and\ \bibinfo {author}
  {\bibfnamefont {Zheshen}\ \bibnamefont {Zhang}},\ }\bibfield  {title}
  {\enquote {\bibinfo {title} {Demonstration of a reconfigurable entangled
  radio-frequency photonic sensor network},}\ }\href {\doibase
  10.1103/PhysRevLett.124.150502} {\bibfield  {journal} {\bibinfo  {journal}
  {Phys. Rev. Lett.}\ }\textbf {\bibinfo {volume} {124}},\ \bibinfo {pages}
  {150502} (\bibinfo {year} {2020})}\BibitemShut {NoStop}%
\bibitem [{\citenamefont {Guo}\ \emph {et~al.}(2020)\citenamefont {Guo},
  \citenamefont {Breum}, \citenamefont {Borregaard}, \citenamefont {Izumi},
  \citenamefont {Larsen}, \citenamefont {Gehring}, \citenamefont {Christandl},
  \citenamefont {Neergaard-Nielsen},\ and\ \citenamefont {Andersen}}]{Guo2020}%
  \BibitemOpen
  \bibfield  {author} {\bibinfo {author} {\bibfnamefont {Xueshi}\ \bibnamefont
  {Guo}}, \bibinfo {author} {\bibfnamefont {Casper~R}\ \bibnamefont {Breum}},
  \bibinfo {author} {\bibfnamefont {Johannes}\ \bibnamefont {Borregaard}},
  \bibinfo {author} {\bibfnamefont {Shuro}\ \bibnamefont {Izumi}}, \bibinfo
  {author} {\bibfnamefont {Mikkel~V}\ \bibnamefont {Larsen}}, \bibinfo {author}
  {\bibfnamefont {Tobias}\ \bibnamefont {Gehring}}, \bibinfo {author}
  {\bibfnamefont {Matthias}\ \bibnamefont {Christandl}}, \bibinfo {author}
  {\bibfnamefont {Jonas~S}\ \bibnamefont {Neergaard-Nielsen}}, \ and\ \bibinfo
  {author} {\bibfnamefont {Ulrik~L}\ \bibnamefont {Andersen}},\ }\bibfield
  {title} {\enquote {\bibinfo {title} {Distributed quantum sensing in a
  continuous-variable entangled network},}\ }\href@noop {} {\bibfield
  {journal} {\bibinfo  {journal} {Nature Physics}\ }\textbf {\bibinfo {volume}
  {16}},\ \bibinfo {pages} {281--284} (\bibinfo {year} {2020})}\BibitemShut
  {NoStop}%
\bibitem [{\citenamefont {Nair}(2018)}]{Nair2018}%
  \BibitemOpen
  \bibfield  {author} {\bibinfo {author} {\bibfnamefont {Ranjith}\ \bibnamefont
  {Nair}},\ }\bibfield  {title} {\enquote {\bibinfo {title} {Quantum-limited
  loss sensing: Multiparameter estimation and bures distance between loss
  channels},}\ }\href {\doibase 10.1103/PhysRevLett.121.230801} {\bibfield
  {journal} {\bibinfo  {journal} {Phys. Rev. Lett.}\ }\textbf {\bibinfo
  {volume} {121}},\ \bibinfo {pages} {230801} (\bibinfo {year}
  {2018})}\BibitemShut {NoStop}%
\bibitem [{\citenamefont {Gramegna}\ \emph {et~al.}(2020)\citenamefont
  {Gramegna}, \citenamefont {Triggiani}, \citenamefont {Facchi}, \citenamefont
  {Narducci},\ and\ \citenamefont {Tamma}}]{GrTrFaNaTa}%
  \BibitemOpen
  \bibfield  {author} {\bibinfo {author} {\bibfnamefont {Giovanni}\
  \bibnamefont {Gramegna}}, \bibinfo {author} {\bibfnamefont {Danilo}\
  \bibnamefont {Triggiani}}, \bibinfo {author} {\bibfnamefont {Paolo}\
  \bibnamefont {Facchi}}, \bibinfo {author} {\bibfnamefont {Frank~A}\
  \bibnamefont {Narducci}}, \ and\ \bibinfo {author} {\bibfnamefont {Vincenzo}\
  \bibnamefont {Tamma}},\ }\bibfield  {title} {\enquote {\bibinfo {title}
  {Heisenberg scaling precision in multi-mode distributed quantum metrology},}\
  }\href@noop {} {\bibfield  {journal} {\bibinfo  {journal} {arXiv preprint
  arXiv:2003.12550}\ } (\bibinfo {year} {2020})}\BibitemShut {NoStop}%
\bibitem [{\citenamefont {Cram{\'e}r}(1999)}]{cramer1999mathematical}%
  \BibitemOpen
  \bibfield  {author} {\bibinfo {author} {\bibfnamefont {Harald}\ \bibnamefont
  {Cram{\'e}r}},\ }\href@noop {} {\emph {\bibinfo {title} {Mathematical methods
  of statistics}}},\ Vol.~\bibinfo {volume} {9}\ (\bibinfo  {publisher}
  {Princeton university press},\ \bibinfo {year} {1999})\BibitemShut {NoStop}%
\bibitem [{\citenamefont {Rao}(1992)}]{rao1992information}%
  \BibitemOpen
  \bibfield  {author} {\bibinfo {author} {\bibfnamefont {C~Radhakrishna}\
  \bibnamefont {Rao}},\ }\bibfield  {title} {\enquote {\bibinfo {title}
  {Information and the accuracy attainable in the estimation of statistical
  parameters},}\ }in\ \href@noop {} {\emph {\bibinfo {booktitle} {Breakthroughs
  in statistics}}}\ (\bibinfo  {publisher} {Springer},\ \bibinfo {year}
  {1992})\ pp.\ \bibinfo {pages} {235--247}\BibitemShut {NoStop}%
\bibitem [{\citenamefont {Pezz\'e}\ and\ \citenamefont
  {Smerzi}(2008)}]{Pezze2008}%
  \BibitemOpen
  \bibfield  {author} {\bibinfo {author} {\bibfnamefont {Luca}\ \bibnamefont
  {Pezz\'e}}\ and\ \bibinfo {author} {\bibfnamefont {Augusto}\ \bibnamefont
  {Smerzi}},\ }\bibfield  {title} {\enquote {\bibinfo {title} {Mach-zehnder
  interferometry at the heisenberg limit with coherent and squeezed-vacuum
  light},}\ }\href {\doibase 10.1103/PhysRevLett.100.073601} {\bibfield
  {journal} {\bibinfo  {journal} {Phys. Rev. Lett.}\ }\textbf {\bibinfo
  {volume} {100}},\ \bibinfo {pages} {073601} (\bibinfo {year}
  {2008})}\BibitemShut {NoStop}%
\bibitem [{\citenamefont {Olivares}\ and\ \citenamefont
  {Paris}(2009)}]{olivares2009}%
  \BibitemOpen
  \bibfield  {author} {\bibinfo {author} {\bibfnamefont {Stefano}\ \bibnamefont
  {Olivares}}\ and\ \bibinfo {author} {\bibfnamefont {Matteo~GA}\ \bibnamefont
  {Paris}},\ }\bibfield  {title} {\enquote {\bibinfo {title} {Bayesian
  estimation in homodyne interferometry},}\ }\href@noop {} {\bibfield
  {journal} {\bibinfo  {journal} {Journal of Physics B: Atomic, Molecular and
  Optical Physics}\ }\textbf {\bibinfo {volume} {42}},\ \bibinfo {pages}
  {055506} (\bibinfo {year} {2009})}\BibitemShut {NoStop}%
\bibitem [{\citenamefont {Berni}\ \emph {et~al.}(2015)\citenamefont {Berni},
  \citenamefont {Gehring}, \citenamefont {Nielsen}, \citenamefont
  {H{\"a}ndchen}, \citenamefont {Paris},\ and\ \citenamefont
  {Andersen}}]{berni2015}%
  \BibitemOpen
  \bibfield  {author} {\bibinfo {author} {\bibfnamefont {Adriano~A}\
  \bibnamefont {Berni}}, \bibinfo {author} {\bibfnamefont {Tobias}\
  \bibnamefont {Gehring}}, \bibinfo {author} {\bibfnamefont {Bo~M}\
  \bibnamefont {Nielsen}}, \bibinfo {author} {\bibfnamefont {Vitus}\
  \bibnamefont {H{\"a}ndchen}}, \bibinfo {author} {\bibfnamefont {Matteo~GA}\
  \bibnamefont {Paris}}, \ and\ \bibinfo {author} {\bibfnamefont {Ulrik~L}\
  \bibnamefont {Andersen}},\ }\bibfield  {title} {\enquote {\bibinfo {title}
  {Ab initio quantum-enhanced optical phase estimation using real-time feedback
  control},}\ }\href@noop {} {\bibfield  {journal} {\bibinfo  {journal} {Nature
  Photonics}\ }\textbf {\bibinfo {volume} {9}},\ \bibinfo {pages} {577--581}
  (\bibinfo {year} {2015})}\BibitemShut {NoStop}%
\bibitem [{\citenamefont {Scully}\ and\ \citenamefont
  {Zubairy}(1997)}]{Scully1997}%
  \BibitemOpen
  \bibfield  {author} {\bibinfo {author} {\bibfnamefont {Marlan~O.}\
  \bibnamefont {Scully}}\ and\ \bibinfo {author} {\bibfnamefont {M.~Suhail}\
  \bibnamefont {Zubairy}},\ }\href {\doibase 10.1017/CBO9780511813993} {\emph
  {\bibinfo {title} {Quantum Optics}}}\ (\bibinfo  {publisher} {Cambridge
  University Press},\ \bibinfo {year} {1997})\BibitemShut {NoStop}%
\bibitem [{\citenamefont {Schleich}(2005)}]{Schleich}%
  \BibitemOpen
  \bibfield  {author} {\bibinfo {author} {\bibfnamefont {Wolfgang~P.}\
  \bibnamefont {Schleich}},\ }\href {\doibase 10.1002/3527602976.ch1} {\emph
  {\bibinfo {title} {Quantum Optics in Phase Space}}}\ (\bibinfo  {publisher}
  {John Wiley and Sons, Ltd},\ \bibinfo {year} {2005})\ \Eprint
  {http://arxiv.org/abs/https://onlinelibrary.wiley.com/doi/pdf/10.1002}
  {https://onlinelibrary.wiley.com/doi/pdf/10.1002} \BibitemShut {NoStop}%
\bibitem [{\citenamefont {Neumann}(1929)}]{Neumann1929}%
  \BibitemOpen
  \bibfield  {author} {\bibinfo {author} {\bibfnamefont {J.~v.}\ \bibnamefont
  {Neumann}},\ }\bibfield  {title} {\enquote {\bibinfo {title} {Beweis des
  ergodensatzes und desh-theorems in der neuen mechanik},}\ }\href {\doibase
  10.1007/BF01339852} {\bibfield  {journal} {\bibinfo  {journal} {Zeitschrift
  f{\"u}r Physik}\ }\textbf {\bibinfo {volume} {57}},\ \bibinfo {pages}
  {30--70} (\bibinfo {year} {1929})}\BibitemShut {NoStop}%
\bibitem [{\citenamefont {von Neumann}(2010)}]{vonNeumann2010}%
  \BibitemOpen
  \bibfield  {author} {\bibinfo {author} {\bibfnamefont {J.}~\bibnamefont {von
  Neumann}},\ }\bibfield  {title} {\enquote {\bibinfo {title} {Proof of the
  ergodic theorem and the h-theorem in quantum mechanics},}\ }\href {\doibase
  10.1140/epjh/e2010-00008-5} {\bibfield  {journal} {\bibinfo  {journal} {The
  European Physical Journal H}\ }\textbf {\bibinfo {volume} {35}},\ \bibinfo
  {pages} {201--237} (\bibinfo {year} {2010})}\BibitemShut {NoStop}%
\bibitem [{\citenamefont {Hiai}\ and\ \citenamefont
  {Petz}(2000)}]{hiai2000semicircle}%
  \BibitemOpen
  \bibfield  {author} {\bibinfo {author} {\bibfnamefont {Fumio}\ \bibnamefont
  {Hiai}}\ and\ \bibinfo {author} {\bibfnamefont {D{\'e}nes}\ \bibnamefont
  {Petz}},\ }\href@noop {} {\emph {\bibinfo {title} {The semicircle law, free
  random variables and entropy}}},\ \bibinfo {number} {77}\ (\bibinfo
  {publisher} {American Mathematical Soc.},\ \bibinfo {year}
  {2000})\BibitemShut {NoStop}%
\bibitem [{\citenamefont {Pucha{\l}a}\ and\ \citenamefont
  {Miszczak}(2017)}]{puchala2017symbolic}%
  \BibitemOpen
  \bibfield  {author} {\bibinfo {author} {\bibfnamefont {Zbigniew}\
  \bibnamefont {Pucha{\l}a}}\ and\ \bibinfo {author} {\bibfnamefont
  {Jaroslaw~Adam}\ \bibnamefont {Miszczak}},\ }\bibfield  {title} {\enquote
  {\bibinfo {title} {Symbolic integration with respect to the haar measure on
  the unitary groups},}\ }\href@noop {} {\bibfield  {journal} {\bibinfo
  {journal} {Bulletin of the Polish Academy of Sciences Technical Sciences}\
  }\textbf {\bibinfo {volume} {65}},\ \bibinfo {pages} {21--27} (\bibinfo
  {year} {2017})}\BibitemShut {NoStop}%
\bibitem [{\citenamefont {Facchi}\ and\ \citenamefont
  {Garnero}(2017)}]{facchi2017quantum}%
  \BibitemOpen
  \bibfield  {author} {\bibinfo {author} {\bibfnamefont {Paolo}\ \bibnamefont
  {Facchi}}\ and\ \bibinfo {author} {\bibfnamefont {Giancarlo}\ \bibnamefont
  {Garnero}},\ }\bibfield  {title} {\enquote {\bibinfo {title} {Quantum
  thermodynamics and canonical typicality},}\ }\href@noop {} {\bibfield
  {journal} {\bibinfo  {journal} {International Journal of Geometric Methods in
  Modern Physics}\ }\textbf {\bibinfo {volume} {14}},\ \bibinfo {pages}
  {1740001} (\bibinfo {year} {2017})}\BibitemShut {NoStop}%
\bibitem [{\citenamefont {Popescu}\ \emph {et~al.}(2006)\citenamefont
  {Popescu}, \citenamefont {Short},\ and\ \citenamefont
  {Winter}}]{popescu2006entanglement}%
  \BibitemOpen
  \bibfield  {author} {\bibinfo {author} {\bibfnamefont {Sandu}\ \bibnamefont
  {Popescu}}, \bibinfo {author} {\bibfnamefont {Anthony~J}\ \bibnamefont
  {Short}}, \ and\ \bibinfo {author} {\bibfnamefont {Andreas}\ \bibnamefont
  {Winter}},\ }\bibfield  {title} {\enquote {\bibinfo {title} {Entanglement and
  the foundations of statistical mechanics},}\ }\href@noop {} {\bibfield
  {journal} {\bibinfo  {journal} {Nature Physics}\ }\textbf {\bibinfo {volume}
  {2}},\ \bibinfo {pages} {754} (\bibinfo {year} {2006})}\BibitemShut {NoStop}%
\end{thebibliography}%

\end{document}